\documentclass[lettersize,journal]{IEEEtran}
\usepackage{amsmath,amsfonts}
\usepackage{algorithmic}
\usepackage{algorithm}
\usepackage{array}
\usepackage[caption=false,font=normalsize,labelfont=sf,textfont=sf]{subfig}
\usepackage{textcomp}
\usepackage{stfloats}
\usepackage{url}
\usepackage{verbatim}
\usepackage{graphicx}
\usepackage{cite}
\usepackage{enumerate}
\usepackage{float} 

\usepackage{amsmath} 
\usepackage{amssymb}
\usepackage{amsthm}
\usepackage[utf8]{inputenc}
\usepackage[english]{babel}
\usepackage{tikzpagenodes} 
\usepackage{tikz} 
\usepackage{pgfplots}

\edef\normalE{\the\mathcode`E}
\begingroup\lccode`~=`E \lowercase{\endgroup
  \def~}{\mathbb{\mathchar\normalE}}
\mathcode`E="8000

\newtheorem{theorem}{Theorem}[section]

\newtheorem{prop}[theorem]{Proposition}
\newtheorem{example}{Example}%
\newtheorem{definition}{Definition}%

\usepackage{multirow}

\begin{document}

\title{The Polynomial Connection between  Morphological Dilation and Discrete Convolution}
\author{
Vivek Sridhar, Keyvan Shahin, Michael Breu{\ss} and Marc Reichenbach 


\thanks{Vivek Sridhar and Micheal Breu{\ss} are with \textit{Institute for Mathematics}, Brandenburg Technical University, Cottbus, Germany. Email: \{sridhviv, breuss\}@b-tu.de}

\thanks{Keyvan Shahin and Marc Reichenbach are with \textit{Institute for Computer Science}, Brandenburg Technical University, Cottbus, Germany. Email: \{keyvan.shahin,  marc.reichenbach\}@b-tu.de}

\thanks{This work was partially funded by the European Regional Development Fund (EFRE 85037495) and BTU Graduate Research School (STIBET short-term scholarship for international PhD Students sponsored by the German Academic Exchange Service (DAAD) with funds of the German Federal Foreign Office) }

}

\markboth{Journal of \LaTeX\ Class Files,~Vol.~14, No.~8, August~2021}%
{Shell \MakeLowercase{\textit{et al.}}: A Sample Article Using IEEEtran.cls for IEEE Journals}

\IEEEpubid{0000--0000/00\$00.00~\copyright~2021 IEEE}

\maketitle

\begin{abstract}
In this paper we consider the fundamental operations 
dilation and erosion of mathematical morphology. 
Many powerful image filtering operations are based on their combinations. 
We establish homomorphism between max-plus semi-ring of 
integers and subset of polynomials over the field of 
real numbers. This enables to reformulate the task of 
computing morphological dilation to that of 
computing sums and products of polynomials. 
Therefore, dilation and its dual operation erosion
can be computed by convolution of discrete linear signals, 
which is efficiently accomplished using a Fast Fourier Transform technique.
The novel method may deal with non-flat filters and 
incorporates no restrictions on shape or size of the 
structuring element, unlike many other fast methods 
in the field. In contrast to previous fast Fourier 
techniques it gives exact results and is not an approximation. 
The new method is in practice particularly suitable 
for filtering images with small tonal range or when 
employing large filter sizes. We explore the benefits by
investigating an implementation on FPGA hardware.
Several experiments demonstrate the exactness and 
efficiency of the proposed method.
\end{abstract}


\section{Introduction} \label{Sec:Intro}


\IEEEPARstart{M}{athematical}  morphology is a highly successful field 
in image processing. 
It is concerned with the analysis of shapes and structures 
in images, see for instance \cite{Serra-Soille,Najman-Talbot,Roerdink-2011} for an account of theory and applications. 
The basic building blocks of many of its processes 
are dilation and erosion. 
These operations are dual, so that it is convenient to focus 
on dilation for the construction of algorithms.
Modeling images via grey values on a discrete grid, 
computing dilation means that a pixel value is set to 
the maximum of the grey values within a filter mask 
centered upon it. 
This mask is called structuring element (SE), 
and it can be either flat or non-flat \cite{Haralick_1}. 
A flat SE describes the shape of the mask, while 
a non-flat SE may additionally contain additive 
grey value offsets.

An important property of morphological filters 
is the high efficiency that can be achieved
in their implementation. Let us briefly review
some efficient algorithms along the lines of
their possible classification described 
in~\cite{Droogenbroeck-Buckley-2005}. 
A first family of schemes aims to reduce the size of the SE 
or to decompose it, thus reducing the number of comparison 
operations for evaluation of the maximum respectively 
minimum over an SE. 
In a second family of methods a given image 
is analysed so that redundant operations 
that may arise in some image parts could be reduced. 
However, most of these methods 
are limited in terms of shape, size or flatness of SE, or specific hardware that is needed,
cf.~\cite{deforges-etal-2013,Lin-Xu-2009,herk-1992,Roerdink-2011,Moreaud,Thurley-Danell-2012}.


There are just few fast methods that allow an SE to be 
of arbitrary shape and size. 
A very popular example is the classic scheme from \cite{Droogenbroeck-Talbot-1996} that relies on histogram updates. However, as also for 
\cite{Droogenbroeck-Talbot-1996}, the algorithmic complexity of most methods relies
inherently on size of the SE, and often also on its shape. 
Since the SE is moved over 
an image in implementations relying on sliding window technique, 
the computational effort also relates to image size. \IEEEpubidadjcol 

An alternative construction of fast algorithms
relies on the possibility to formulate operations over an SE as convolutions, 
which may be realized via a fast transform.
In a first work \cite{TMG}, binary dilation 
respectively erosion are represented by convolution of
characteristic functions of underlying sets. 
In \cite{Kukal} this approach was extended in a straightforward way to
grey scale images. This was done by decomposing an image into its level sets,
and each level set was processed like a binary image. By construction,
the method is limited to flat filters of particular shape.
A different extension of \cite{TMG} has been proposed in \cite{VMB_1}, 
making use of an analytical approximation of morphological operations.
The resulting method is suitable for flat and non-flat SE, without restriction on
shape or size. However, as analyzed in \cite{VMB_1,VMB_2}, this comes
at the expense of a shift and smoothing effect 
in the tonal histogram. In order to address this issue, a novel fast and exact method has recently 
been proposed \cite{SB-2022-1} which considers the umbra of  image and filter as the computational setting for the convolution. 

\subsubsection*{Our Contributions} 

In this paper we extend the work in \cite{SB-2022-1} to analytically prove the exactness of the proposed method. In particular, we construct a homomorphism between semi-rings of polynomials and max-plus semi-ring of non-negative integers. This allows us to represent computations in the max-plus semi-ring as sums and products of polynomials. We also establish how our constructions and convolution precisely relate to morphological dilation, thus allowing us to utilize the above-mentioned theory. 
Furthermore we provide extensive additional experiments that demonstrate 
the exactness of the proposed method. In particular, this includes 
a novel, detailed study of an implementation of our method 
on FPGA hardware. The results confirm the beneficial theoretical 
and computational properties of our scheme.

\section{Theoretical Background}  \label{Sec:2}

\subsection{Morphological Operations} \label{Sec:Morph}

An $N$-dimensional grey-value image is a function 
$f:F\rightarrow \mathbb{L}$. 
Thereby $F\subseteq \mathbb{Z}^N$ is the set of the (in general, $N$-dimensional) indices 
of pixels in the image, also known as domain of the image. 
Furthermore, $\mathbb{L}$ $=\{0,1,\cdots l\}$,  $l>0$, is the tonal range of the image. 
In case of the common 8-bit grey-value image, $\mathbb{L}$ is the set of integers in the range $[0,255]$.
Similarly, a morphological grey-value filter, flat or non-flat, can be defined 
as $b:B\rightarrow \mathbb{L}$, $B\subseteq \mathbb{Z}^N$. 
Let us note that $b(\cdot)\geq 0$ (which affects some formula).
The mask domain $B$ denotes the domain of the structuring element. 

The dilation of image $f$ by structuring elements $b$ is denoted by $f\oplus b$ and is computed for each $x$ in its domain $F\oplus B$ as 
\begin{equation}
    \label{Equation:6new}
    (f\oplus b)(x)=
    \max\{f(x-y)+b(y) \, \vert \, (x-y)\in F \land y\in B\} 
\end{equation}
Here, $F\oplus B$ $=\{x_F+ x_B \vert x_F \in F \land x_B\in B \}$. $F\oplus B$ is the Minkowski addition of set of indices. This also corresponds to dilation of $F$ by $B$ considering them as binary images (see, \cite{Haralick_1}). In practice, we are only interested in the indices contained the original image, therefore, for each $x\in F$, we use,
with $f\oplus b :=(f\oplus b)(x) $:
\begin{equation}
    \label{Equation:6}
    \hspace{-2.0ex}f\oplus b=\begin{cases}
    0 \quad \text{ if } x\not\in F\oplus B \vspace{1.0ex}\\
    \max\{f(x-y)+b(y) \, \vert \, (x-y)\in F \land y\in B\} \\
    \qquad \text{otherwise}
    \end{cases}
\end{equation}

\begin{example}
\label{Example:2}
(Adopted from \cite{SB-2022-1}.)
Consider a 1-dimensional image $f=\begin{bmatrix} \underline{3} & 0 & 7 &6 &2 &7\end{bmatrix}$ 
and filter $b=\begin{bmatrix}1 &\underline{2} &X &0\end{bmatrix}$. Here, $\underline{a}$ denotes 
that $a$ is exactly the entry that is located at the index $0$, i.e. at the spatial origin. 
The symbol $X$ is used for indices not in the domain. 
 
From \eqref{Equation:6}, we have  
$(f\oplus b)(0)$ $=\max\{f(0)+b(0),f(1)+b(-1)\}$ $=\max\{3+2, 1+0\}$  $=5$. 
Similarly, $(f\oplus b)(2)$ $=\max\{f(0)+b(2),f(2)+b(0), f(3)+b(-1)\}$ $=\max\{3,9,7 \}$  $=9$ 
and so on. 

Thus, we get, $(f\oplus b)=$ $\begin{bmatrix} \underline{5} &8&9&8&8&9\end{bmatrix}$.
\end{example}

The other fundamental morphological operation is erosion. Erosion of image $f$ by structuring element $b$ is denoted as $f\ominus b$ and is computed for each $x$ in its domain $F\ominus B$ as
\begin{equation}
    \label{Equation:Erosion}
    (f\ominus b)(x)=
    \min\{f(x+y)-b(y) \, \vert \, x+y\in F \land y\in B\} 
\end{equation}

\begin{figure}[h]
\centering
\minipage{0.33\linewidth}
\centering
 \includegraphics[width=\linewidth]{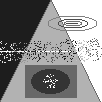}
\endminipage\hfill
\minipage{0.33\linewidth}
 \includegraphics[width=\linewidth]{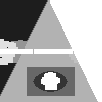}
\endminipage\hfill
\minipage{0.33\linewidth}
 \includegraphics[width=\linewidth]{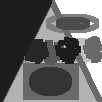}
\endminipage\hfill
\caption{\label{sample_dial_erod} 
{\bf: Left} \textit{Sample image} of size $99\times 99$. {\bf Centre:}  Dilation. {\bf Right:} Erosion, with a $5\times 5$ flat filter. 
}

\end{figure}

Grey-value dilation and erosion are \textit{dual} operations, i.e.\ $f\oplus b = -((-f)\ominus \breve{b})$. Here, $\breve{b}$ is the reflection of structuring element about origin (of $Z^N$) and $(-f)$ denotes the negative of image,  i.e.\ $(-f)(x) = l-f(x)$, $\forall x\in F$ (see, \cite{Viv_1}). Due to the duality, the task of computing grey-value erosion can be reduced to that of computing grey-value dilation in linear time. Thus, in this paper, we focus on computing dilation. 

\subsection{Polynomials and Max-plus semi-ring} \label{Sec:MaxPlus}

The connection between grey-value morphological dilation and the \textit{tropical semi-ring} $(\mathbb{R} _{\max}, \max, +)$,  here  $\mathbb{R}_{\max} = \mathbb{R} \cup \{ -\infty \}$, has been previously explored in \cite{Maragos_2019}. In particular, 
the computation at each pixel of dilation for grey-value images using 
a non-flat structuring element, is of form $\max \{a_1+b_1, a_2+b_2 \cdots a_m+b_m\}$ (see Example \ref{Example:2}). This computation takes place in the semi-ring $(\mathbb{Z}^{+\text{ve}} _{\max}, \max, +)$,  as $\mathbb{L} \subset$  $\mathbb{Z}^{+\text{ve}}_{\max}$ $= \mathbb{N}\cup \{0\} \cup \{-\infty\}$. Note, that the pixel not in the domain of image or SE can be assumed to have the value $-\infty$.

The theory that we develop in this section allows us to reduce the problem of computation in semi-ring $(\mathbb{Z}^{+\text{ve}} _{\max}, \max, +)$ to problem of computing sums and products of polynomials over the real field $(\mathbb{R}, +, \cdot)$.
Our construction in the next section (Section \ref{Sec:Method}), allows us to compute sums and products of polynomials as convolution of real numbers (or integers). Thus, we are able to utilize Fast Fourier Transforms (or Fast Number Theoretic Transforms) to speed up the computations. 

We begin by recalling the definition of the algebraic structure \textit{semi-ring}, which is ring without an additive inverse (see, Chapter 1 of \cite{Golan_2013}). 

\begin{definition}{\textbf{Semi-ring}}
\label{Def:Semiring}
A semi-ring $(R, +, \cdot)$ is a non-empty set $R$, equipped with two binary operators $+$ and $\cdot$ such that they satisfy the following properties:
\begin{enumerate}[i.]
    \item \textbf{Closure of addition:} $a+b\in R$, $\forall a$, $b$ $\in R$.
    \item \textbf{Associativity of addition:} $(a+b)+c=a+(b+c)$, $\forall a,b,c \in R$
    \item \textbf{Existence of additive identity:} $\exists 0\in R$ such that $a+0=0+a=a$, $\forall a\in R$.
    \item \textbf{Commutativity of addition:} $a+b=b+a$, $\forall a,b\in R$.
    \item \textbf{Closure of multiplication:} $a\cdot b \in R$, $\forall a,b\in R$.
    \item \textbf{Associativity of multiplication:} $a (b c)=(a b) c$, $\forall a,b,c \in R$.
    \item \textbf{Existence of multiplicative identity:} $\exists 1\in R$ such that $a 1=1 a=a$, $\forall a\in R$
    \item \textbf{Distributive laws:} Multiplication left and right distributes over addition, i.e.\ for all $a,b,c,\in R$, we have,
    \begin{enumerate}
        \item \textbf{Left distribution} $a(b+c)=ab+ac$.
        \item \textbf{Right distribution} $(a+b)c=a c+b c$.
    \end{enumerate}
    \item $a\cdot 0 =0\cdot a =0$, $\forall a\in R$.
\end{enumerate}
In addition to the above properties, if $a b= b a$, $\forall a,b\in R$, then the semi-ring is said to be \textit{commutative}.
\end{definition}

We now establish that  $(\mathbb{Z}^{+\text{ve}} _{\max}, \max, +)$ is a semi-ring.

\begin{prop}
\label{p:Zsemiring}
$(\mathbb{Z}^{+\text{ve}} _{\max}, \max, +)$ is a commutative semi-ring.
\end{prop}

\begin{proof}
We know that $(\mathbb{R}_{\max}, \max, +)$ forms a semi-ring known as \textit{tropical semi-ring} \cite{Maragos_2019}.

Clearly, $\mathbb{Z}^{+\text{ve}}_{\max}$ $= \mathbb{N}\cup \{0\} \cup \{-\infty\}$ $\subset \mathbb{R}_{\max}$ $= \mathbb{R}\cup \{-\infty\}$. To establish that $\mathbb{Z}^{+\text{ve}}_{\max}$ is a semi-ring (i.e.\ sub-semiring of $(\mathbb{R}_{\max}, \max, +)$), we show the closure of operators and existence of identities in $\mathbb{Z}^{+\text{ve}}_{\max}$ (see, Chapter 1 of \cite{Golan_2013}).

\begin{enumerate}[i.]
    \item By definition of $\mathbb{Z}^{+\text{ve}}_{\max}$, the identity of operator '$+$', is $0\in \mathbb{Z}^{+\text{ve}}_{\max}$ and the identity of operator '$\max$', is $-\infty \in \mathbb{Z}^{+\text{ve}}_{\max}$. 
    \item We prove the closure of operators. 
    
    As for the whole numbers, $\mathbb{N}\cup\{0\}$ is closed under operator '$+$'. Therefore, for any $a,b\in \mathbb{N}\cup\{0\}$, $a+b\in \mathbb{N}\cup\{0\} \subset \mathbb{Z}^{+\text{ve}}_{\max}$. 
    
    If $a= -\infty$ or $b=-\infty$, then $a+b=-\infty\in \mathbb{Z}^{+\text{ve}}_{\max}$.
    Thus, $\mathbb{Z}^{+\text{ve}}_{\max}$ is closed under '$+$'.
    
    Similarly, for any $a,b\in \mathbb{N}\cup\{0\}$, $\max \{a,b\} \in \mathbb{N}\cup\{0\}$.  
    If $a=-\infty$, then $\max \{a,b\} =b$ and if $b=-\infty$, then $\max \{a,b\} =a$. Thus, $\mathbb{Z}^{+\text{ve}}_{\max}$ is closed under '$\max$'.

\end{enumerate}

Since $(\mathbb{R}_{\max}, \max, +)$ is commutative semi-ring, $(\mathbb{Z}^{+\text{ve}} _{\max}, \max, +)$ is also commutative.
\end{proof}

Consider the ring of polynomials, $(\mathbb{R}[x], +, \cdot)$, over the field of real numbers, $(\mathbb{R}, +, \cdot)$.  $\mathbb{R}[x]$ consists of all polynomials, in a single variable $x\in\mathbb{R}$, with real coefficients and non-negative integers as powers. Let $\mathbb{P}\subset \mathbb{R}[x]$, be the set of polynomials with non-negative real coefficients, e.g.\ $0$, $3x^2+2$, $x$, $1.5$, $7.3x+9$ etc.\ 

\begin{prop}
\label{p:Psemiring}
$(\mathbb{P}, +, \cdot)$ is a commutative semi-ring.
\end{prop}

\begin{proof}
$(\mathbb{R}[x], +, \cdot)$ is a commutative ring \cite{Herstein_1975}. To prove the proposition, we show that $(\mathbb{P},+,\cdot)$ is a sub-semiring of $(\mathbb{R}[x], +, \cdot)$. We need the closure of operators and existence of identities in $\mathbb{P}$ (see, Chapter 1 of \cite{Golan_2013}).

\begin{enumerate}[i.]
    \item The additive and multiplicative identity of $(\mathbb{R}[x], +, \cdot)$ are $0$ and $1$, respectively. By definition of $\mathbb{P}$, $0\in \mathbb{P}$ and $1\in \mathbb{P}$.
    \item The set of non-negative real numbers, $\{x\vert x\in \mathbb{R}, x\geq 0\}$ is closed under addition and multiplication. The set of non-negative integers, $\mathbb{N}\cup\{0\}$ is closed under addition. Thus, for any $p_1, p_2 \in P$, $p_1 +p_2 \in \mathbb{P}$, as coefficients of $p_1+p_2$ are non-negative real numbers. Similarly, $p_1\cdot p_2 \in \mathbb{P}$, as coefficients of $p_1\cdot p_2$ are non-negative and powers are non-negative integers. 
\end{enumerate}
Since, $(\mathbb{R}[x], +, \cdot)$ is commutative, $(\mathbb{P}, +, \cdot)$ is a commutative semi-ring.
\end{proof}

Let $\delta : \mathbb{R}[x] \rightarrow \mathbb{Z}^{+\text{ve}} _{\max}$ be the degree function, i.e.
\begin{equation}
    \label{Eq:Degree}
    \delta (p) = \begin{cases}
    -\infty &\text{if } p=0\\
    \text{highest power with} \\ \text{non-zero coefficient}
   &\text{otherwise.}
    \end{cases}
\end{equation}
We know that (see, Chapter 3 of \cite{Herstein_1975}):
\begin{enumerate}[i.]
    \item $\delta(p_1 \cdot p_2) =\delta(p_1) +\delta(p_2)$, $\forall p_1, p_2 \in \mathbb{R}[x]$.
    \item $\delta(p_1+p_2) \leq \max \{\delta(p_1), \delta(p_2)\}$, $\forall p_1, p_2 \in \mathbb{R}[x]$.
\end{enumerate}

We now show that $\delta(.)$ is a homomorphism from $\mathbb{P}$ onto $\mathbb{Z}^{+\text{ve}} _{\max}$.

\begin{prop}
\label{p:DegHomomorphism}
$\delta(.) : \mathbb{P} \rightarrow \mathbb{Z}^{+\text{ve}} _{\max}$ is a homomorphism from semi-ring $(\mathbb{P}, +, \cdot)$ onto semi-ring $(\mathbb{Z}^{+\text{ve}} _{\max}, \max, +)$.
\end{prop}

\begin{proof}

We first show that $\delta(.) : \mathbb{P} \rightarrow \mathbb{Z}^{+\text{ve}} _{\max}$ is a surjection.

For $-\infty \in \mathbb{Z}^{+\text{ve}} _{\max}$ , we have $0\in \mathbb{P}$ such that $\delta(0)=-\infty$. 
Similarly, for $0\in  \mathbb{Z}^{+\text{ve}} _{\max}$, we have $1\in \mathbb{P}$, such that $\delta (1)=0$. For any $a\in \mathbb{N}$, we have $x^a\in \mathbb{P}$, such that $\delta(x^a)=a$.

For $\delta(.)$ to be a homomorphism between the semi-rings, it needs to map the identities and preserve the operations (see, Chapter 9 of \cite{Golan_2013}).

\begin{enumerate}[i.]
    \item We have $\delta(0)=-\infty$ and $\delta(1)=0$. Therefore, $\delta(.)$ maps the additive and multiplicative identities of $(\mathbb{P}, +, \cdot)$ to the corresponding identities of $(\mathbb{Z}^{+\text{ve}} _{\max}, \max, +)$.
    \item We already have established $\delta(p_1\cdot p_2) =\delta(p_1)+\delta(p_2)$, $\forall p_1, p_2 \in \mathbb{P}$.
    
    We also know $\delta(p_1+p_2) \leq \max \{\delta(p_1)+\delta(p_2)\}$, $\forall p_1, p_2 \in \mathbb{P}$. We prove equality. 
    If $p_1 =0$, then $\delta(p_1+p_2)=\delta(p_2)$ $=\max\{\delta(p_1),\delta(p_2)\}$, as $\delta(p_1)=-\infty$. The case $p_2 =0$ works analogously.
    
    Let $p_1 \neq 0$ and $p_2 \neq 0$. Then, $\delta(p_1)=a$ and $\delta(p_2)=b$, for some $a,b\in \mathbb{N}\cup \{0\}$. Since we are dealing with commutative semi-rings, without loss of generality, we can assume $a\geq b$.
    Since coefficients of $p_1$ and $p_2$ are non-negative, coefficient of $x^a$ is non-zero in $p_1+p_2$. Therefore, $\delta(p_1+p_2)\geq a=\max\{\delta(p_1),\delta(p_2)\}$. 
    This means $\delta(p_1+p_2)=\max\{\delta(p_1),\delta(p_2)\}$
    
    $\therefore \delta(p_1+p_2)=\max\{\delta(p_1),\delta(p_2)\}$, $\forall p_1, p_2\in \mathbb{P}$. \qedhere

\end{enumerate}


\end{proof}

We define an injection, $\delta '(.):\mathbb{Z}^{+\text{ve}} _{\max } \rightarrow \mathbb{P}$, which computes an inverse of $\delta(.):\mathbb{P}\rightarrow \mathbb{Z}^{+\text{ve}} _{\max } $

\begin{equation}
    \label{eq:InvDelta}
    \delta '(a) =\begin{cases}
        0 &\text{if } a=-\infty\\
        1 &\text{if } a=0\\
        x^a &\text{, } \forall a\in \mathbb{N}
    \end{cases}
\end{equation}

Clearly, one can observe that:
\begin{equation}
 \label{eq:deltaInvDelta}
   \delta(\delta '(a))=a \text{, }\forall a\in \mathbb{Z}^{+\text{ve}} _{\max }. 
\end{equation}


\begin{prop}
\label{p:InvDelta}
 $\delta '(.):\mathbb{Z}^{+\text{ve}} _{\max } \rightarrow \mathbb{P}$ is an injection.
\end{prop}

\begin{proof}
Let $\delta '(a_1)=\delta '(a_2)$, for any $a_1,a_2\in \mathbb{Z}^{+\text{ve}} _{\max }$.

\noindent 
If $\delta '(a_1)=\delta '(a_2) =0$, then $a_1=a_2=-\infty$, by definition.

 \noindent If $\delta '(a_1)=\delta '(a_2) =1$, then $a_1=a_2=0$, by definition.

 \noindent Otherwise, $\delta '(a_1)=\delta' (a_2)$ $\Rightarrow x^ {a_1} =x^{a_2}$ $\Rightarrow a_1=a_2$.
\end{proof}

The next theorem is the main result of this section, which allows us to describe computations in $(\mathbb{Z}^{+\text{ve}} _{\max}, \max, +)$ semi-rings in terms of sums and products of polynomials. 

\begin{theorem}
\label{t:MaxPluPoly}
The following equality is true:
\begin{equation}
\label{eq:Thm} 
\begin{split}
\max \{a_1+b_1, a_2+b_2 \cdots a_k+b_k\} \;
= \; \delta(\sum _{i=i} ^{m} p_{a_i} \cdot p_{b_i})
\end{split}
\end{equation}
where,
\begin{enumerate}[i.]
    \item $m \in \mathbb{N}$ and $m\geq 2$
    \item $a_i, b_i \in \mathbb{Z}^{+\text{ve}} _{\max}$, and
    \item $p_a=\delta '(a)$, for $a\in \mathbb{Z}^{+\text{ve}} _{\max}$
\end{enumerate}
\end{theorem}

\begin{proof}
We prove the theorem using the principle of mathematical induction. 
We make use of the semi-ring structures of $(\mathbb{Z}^{+\text{ve}} _{\max}, \max, +)$ and $(\mathbb{P}, +, \cdot)$ (see, Proposition \ref{p:Zsemiring} and Proposition \ref{p:Psemiring}), especially the associativity of operators (see, Definition \ref{Def:Semiring}) and the homomorphism, $\delta(.)$, between the semi-rings (Proposition \ref{p:DegHomomorphism}).

For $m=2$, we have,
\begin{equation*}
    \begin{split}
    \lefteqn{\max \{a_1+b_1, a_2+b_2\}}\\
    & =  \max\{\delta (\delta ' (a_1)) +\delta (\delta ' (b_1)), \delta (\delta ' (a_2)) +\delta (\delta ' (b_2)) \} \\
    & = \max \{\delta(p_{a_1})+\delta(p_{b_1}),\delta(p_{a_2})+\delta(p_{b_2})\} \\
    & =\max \{\delta (p_{a_1} \cdot p_{b_1}) , \delta (p_{a_2} \cdot p_{b_2}) \}\\
    &= \delta ((p_{a_1} \cdot p_{b_1} )+( p_{a_2} \cdot p_{b_2}))
    \end{split}
\end{equation*}
Thus, Equation \eqref{eq:Thm} holds for $m=2$. 

Let the statement in Equation \eqref{eq:Thm} hold for $m=k$, $k\in \mathbb{N}, k\geq2$.

We have,
\begin{equation*}
    \max \{ a_1 +b_1, a_2+b_2, \cdots, a_k+b_k\} =\delta(\sum _{i=1}^{k} p_{a_i}\cdot p_{b_i})
\end{equation*}

Let $a_0=\max  \{ a_1 +b_1, a_2+b_2, \cdots, a_k+b_k\}$.  Similarly, let $p_0 =\sum _{i=1}^{k} p_{a_i}\cdot p_{b_i}$. Clearly, due to closure of operators in semi-rings, $a_0 \in \mathbb{Z}^{+\text{ve}} _{\max}$ and $p_0\in \mathbb{P}$.

We now show that Equation \eqref{eq:Thm} holds for $m=k+1$.
For any $a_{k+1},b_{k+1} \in \mathbb{Z}^{+\text{ve}} _{\max}$, we have 
\begin{equation*}
    \begin{split}
        \lefteqn{\max \{a_1 +b_1, a_2+b_2, \cdots, a_k+b_k, a_{k+1}+b_{k+1}\}}\\
        & =  \max \{\max\{a_1 +b_1, a_2+b_2, \cdots, a_k+b_k\}, a_{k+1}+b_{k+1}\}\\
        & = \max \{a_0, a_{k+1}+b_{k+1} \}\\
        & = \max \{\delta (\delta '(a_0)), \delta(\delta'(a_{k+1}))+\delta(\delta '(b_{k+1})) \} \text{, by \eqref{eq:deltaInvDelta}}\\
        & = \delta(p_{a_0}+(p_{a_{k+1}}\cdot p_{b_{k+1}}))
    \end{split}
\end{equation*}

We have $\delta (p_0)=a_0=\delta(\delta '(a_0))=\delta (p_{a_0})$. Thus, we get:
\begin{equation*}  
\begin{split}
    \lefteqn{\delta(p_{a_0}+(p_{a_{k+1}}\cdot p_{b_{k+1}}))} \\
    &= \max\{\delta(p_{a_0}), \delta(p_{a_{k+1}}\cdot p_{b_{k+1}}) \} \\
    &=\max\{\delta(p_0), \delta(p_{a_{k+1}}\cdot p_{b_{k+1}}) \} \\
    &= \max \{\delta (\sum _{i=1}^{k} p_{a_i}\cdot p_{b_i}),  \delta(p_{a_{k+1}}\cdot p_{b_{k+1}})\}\\
    &=\delta((\sum _{i=1}^{k} p_{a_i}\cdot p_{b_i}) + (p_{a_{k+1}}\cdot p_{b_{k+1}}))\\
    &= \delta(\sum _{i=1}^{k+1} p_{a_i}\cdot p_{b_i})
\end{split}
\end{equation*}

i.e.
\begin{equation*}
    \begin{split}
         \max \{a_1 +b_1, a_2+b_2, \cdots, a_k+b_k, a_{k+1}+b_{k+1}\} \\
         =\delta(\sum _{i=1}^{k+1} p_{a_i}\cdot p_{b_i}) 
    \end{split}
\end{equation*}

Here, we have shown that,  Equation \eqref{eq:Thm}  holds for $m=k+1$, if it holds for $m=k$.  
Thus, Equation \eqref{eq:Thm} holds for $m\in \mathbb{N}$, $m\geq 2$. 
\end{proof}

\subsection{Discrete Linear Convolutions}
\label{Sec:Convolutions}
In this section we briefly discuss some basic notions and properties of linear discrete convolution which are utilized in our proposed method. 
The content of this section is extracted from part of \cite{SB-2022-1},
focusing on the techniques that are relevant for the extensions in 
the current paper.

\subsubsection*{One-Dimensional Discrete Linear Convolution}
Consider two 1-dimensional discrete signals $f:F\rightarrow \mathbb{R}$ and 
$g:G\rightarrow \mathbb{R}$, where $F,G\subseteq \mathbb{Z}$. The convolution of $f$ and $g$, 
denoted by $f\circledast g$, results in a 1-dimensional discrete signal 
$h:\mathbb{Z}\rightarrow \mathbb{R}$, computed by:
\begin{equation}
\label{Equation:1}
\begin{split}
    h[k] &= (f\circledast g)[k] =(g\circledast f)[k] \\
    &= \sum _{i=-\infty} ^{\infty} f[i]g[k-i] \text{  ,} \quad \forall k \in \mathbb{Z}
\end{split}
\end{equation}
In \eqref{Equation:1}, $f$ and $g$ are sufficiently 
\textit{padded} with $0$s, i.e, $f[i]=0$, if $i\not\in F$, and  $g[i]=0$, if $i\not\in G$.

If $F$ and $G$ are finite sub-intervals of $\mathbb{Z}$, we might be interested in the 
output $h=(f\circledast g)$, only over a finite subset $H\subseteq \mathbb{Z}$. 
This subset is determined by the \textit{mode} of convolution \cite{SciPy_Convolve}.
In this work, we utilize the \textit{full mode} of convolution, in which the output $h_{full} = (f\circledast _{full} g)$ omits all the elements of the linear discrete convolution whose computation only involves \textit{padded} parts of the inputs:

\begin{equation}
    \label{Equation:2}
    \begin{split}
    h_{full}[k] &= (f\circledast _{full}g)[k]\\
    &=\sum _{i \, : \, i\in F \land (k-i)\in G} f[i]g[k-i]
    \end{split}
\end{equation}

%

\begin{example}
\label{example:1}
Let us clarify the meaning of the notions from above by giving explicit formulae for 
corresponding convolutions of two signals of finite lengths $n_1+1$ respectively $n_2+1$. 
Consider two signals $f:[0,n_1]\rightarrow \mathbb{R}$ and $g:[0,n_2]\rightarrow \mathbb{R}$. 
The linear discrete convolution of the two signals is given by 
\begin{equation}
\label{Equation:3}
    h[k]=(f\circledast g)[k]= 
    \begin{cases}
    \displaystyle{\sum _{i=\max\{0,k-n_2\}} ^{\min\{n_1,k\}}} f[i]g[k-i] \\
    \qquad \text{if } 0\leq k\leq n_1+n_2 \vspace{1.0ex}\\
    0 \quad \text{ otherwise}
    \end{cases}
\end{equation}
The \textit{full mode} of convolution is given by
\begin{equation}
    \label{Equation:4}
    \begin{split}
    (f\circledast_{full}g)[k]=\sum_{i=\max\{0,k-n_2\}}^{\min \{n_1,k\}} f[i]g[k-i] \text{, } \\
     \forall k\in [0,n_1+n_2] 
    \end{split}
\end{equation}
\end{example}

\begin{example}
\label{e:4}
Let us consider another example of full convolution, involving discrete $1$-dimensional signals with negative indices. Let $f_1:[-1, 7]\rightarrow \mathbb{R}$ and   $g_1:[-3,5]\rightarrow \mathbb{R}$.
\begin{equation}
    \label{Equation:e4}
    \begin{split}
    (f_ 1\circledast_{full}g_1)[k]=\sum_{i=\max\{-1,k-5\}}^{\min \{7,k+3\}} f_1[i]g_1[k-i] \text{, } \\
     \forall k\in [-4,12] 
    \end{split}
\end{equation}
\end{example}

\subsubsection*{Multi-Dimensional Linear Discrete Convolution}
Multi-dimensional linear discrete convolution is a straightforward extension of 
$1$-dimensional convolution. Let $\tilde{f}:\tilde{F}\rightarrow R$ and 
$\tilde{g}:\tilde{G}\rightarrow R$ be $N$-dimensional discrete signals, 
i.e. $\tilde{F}, \tilde{G} \subseteq \mathbb{Z}^N$. 
Setting $\theta_j := k_j-i_j$,
the convolution 
$(\tilde{f}\circledast \tilde{g})=\tilde{h}:\mathbb{Z}^N\rightarrow \mathbb{R}$ is defined as:
\begin{equation}
\label{Equation:MultiDimensionalConv}
\begin{split}
    \tilde{h}[k_1,k_2,\ldots, k_N] 
    &= (\tilde{f}\circledast \tilde{g})[k_1,k_2,\ldots, k_N] \\ 
     = \sum _{i_1=-\infty} ^{\infty} & \ldots  \sum _{i_N=-\infty} ^{\infty}  \tilde{f}[i_1,\ldots, i_N]\tilde{g}[\theta_1,\ldots, \theta_N]
\end{split}
\end{equation}
valid for all $k_1, k_2, \ldots k_N \in \mathbb{Z}$ within all the $\theta_j$, and 
where $\tilde{f}$ and $\tilde{g}$ are sufficiently \textit{padded}.

If $\tilde{F}$ and $\tilde{G}$ are finite and $N$-dimensional rectangles, i.e. $\tilde{F}=F_1\times F_2\times \ldots F_N$, $\tilde{G}=G_1\times G_2\times \ldots G_N$,
where all the $F_1,\ldots, F_N,G_1,\ldots, G_N$ are sub-intervals of $\mathbb{Z}$, the finite domain of interest of output  $\tilde{H}\subseteq \mathbb{Z}^N$  is specified by the \textit{mode} of convolution along each dimension, similar to $1$-dimensional case. 

\section{Proposed Method}
\label{Sec:Method}
Let us first sketch the general proceeding of our novel method.
We consider in a general setting an $N$-dimensional grey-value 
image $f:F\rightarrow \mathbb{L}$ and filter 
(flat or non-flat) $b:B\rightarrow \mathbb{L}$, $F,B\subseteq \mathbb{Z}^N$.
First we construct $(N+1)$-dimensional arrays, $f_{Um}$ and $b_{Um}$, 
for $f$ and $b$ respectively, such that, we have a $1$-dimensional vector which represents the coefficients of the polynomial ($\delta '(f(x))$ or $\delta ' (b(x))$) corresponding to each pixel ($x)$. The constructions $f_{Um}$ and $b_{Um}$ are similar to the classic notion of the umbra of an image, cf.~\cite{Haralick_1},  to which we make a reference with the subscript.

By Theorem \ref{t:MaxPluPoly} (and Equation \eqref{eq:Thm}), 
the constructed arrays allow us to transfer the problem of morphological dilation 
(computing maxima of sum)
in $N$ dimensions to that of convolution in $(N+1)$ dimensions. 
There we can use an efficient method such as the FFT to compute the 
convolution, $(f\oplus b)_{Um}$, of $f_{Um}$ and $b_{Um}$. 
The dilated image, $(f\oplus b)$, can be obtained by appropriately 
projecting the $(N+1)$-dimensional array $(f\oplus b)_{Um}$ on $N$-dimensions, i.e.\ by computing $\delta(p_x)$ at each $x\in F$, where $p_x$ is the polynomial whose coefficients constitute $(f\oplus b)_{Um}(x)$.

The detailed description of our proposed method is given as follows.

\subsubsection*{Step 1}\label{Step:1}
Let $l_R=\max_{x\in F}\{f(x)\} + \max _{x\in B} \{ b(x)\}$. 
Construct two $(N+1)$-dimensional arrays  $f_{Um}$ and $b_{Um}$. 
The first $N$ dimensions, referred to as the \textit{domain-dimensions}, 
of $f_{Um}$ and $b_{Um}$ consists of all ($N$-dimensional) indices of $F$ and $B$, respectively. 
The last dimension, referred to as the \textit{range-dimension} consists of 
indices $\{0,1,\ldots l_R\}$. 

The \textit{range-dimension}, $f_{Um}(x)$ (or $b_{Um}(x)$), represents the coefficients of polynomial $p_x=\delta '(f(x))$ (or $p_x=\delta '(b(x))$) corresponding to the pixel value of image (or filter) at position $x$. The values in $f_{Um}$ and $b_{Um}$ are determined 
by the following two equations:
\begin{equation}
\label{Equation:7}
    f_{Um}(x,y)=\begin{cases}
    1 &\text{if } x\in F \text{ and } f(x)=y\\
    0 &\text{otherwise.}
    \end{cases}
\end{equation}
\begin{equation}
\label{Equation:8}
    b_{Um}(x,y)=\begin{cases}
    1 &\text{if } x\in B \text{ and } b(x)=y\\
    0 &\text{otherwise.}
    \end{cases}
\end{equation}
Note that the indices of \textit{range-dimension} start from $0$. 
The above construction makes it possible to have image and filter of any shape and including \textit{gaps} in the domain. 

We pad or fill the domain appropriately so that $F$ and $B$ are (hyper-) rectangles. For example, if $x_0 \not \in B$, then, we can define $b(x_0)=-\infty$, i.e.\ $\delta '(b(x_0))=\delta '(0)=0$, thus,  $b_{Um}(x_0,y)=0$, $\forall y\in \{0,1,\ldots l_R\}$. 
Thus, the constructed arrays $f_{Um}$ and $b_{Um}$ will always be $(N+1)$-dimensional (hyper-) rectangles in shape, regardless of the shape of image domain $F$ and filter domain $B$.

\begin{example}
\label{Example:3}
(Adopted from \cite{SB-2022-1}.)
Consider the image $f$ and filter $b$ in Example \ref{Example:2}.  
Then $f_{Um}$ and $b_{Um}$ are $2$-dimensional arrays (meaning that we have matrices here), 
with the column (\textit{range dimension}) having indices $0,1,\ldots 9$, since $f$ has the range of values in $[0,7]$ and $b$ in $[0,2]$.
Since the image $f$ in Example \ref{Example:2} has six pixels, $f_{Um}$ has six corresponding columns.

\begin{center}

$f_{Um}=\begin{bmatrix}
0 &0 &0 &0 &0 &0 \\
0 &0 &0 &0 &0 &0 \\
0 &0 &1 &0 &0 &1 \\
0 &0 &0 &1 &0 &0 \\
0 &0 &0 &0 &0 &0 \\
0 &0 &0 &0 &0 &0 \\
1 &0 &0 &0 &0 &0 \\
0 &0 &0 &0 &1 &0 \\
0 &0 &0 &0 &0 &0 \\
\underline{0} &1 &0 &0 &0 &0 
\end{bmatrix}$
$b_{Um}=\begin{bmatrix}
0 &0 &0 &0 \\
0 &0 &0 &0 \\
0 &0 &0 &0 \\
0 &0 &0 &0 \\
0 &0 &0 &0 \\
0 &0 &0 &0 \\
0 &0 &0 &0 \\
0 &1 &0 &0 \\
1 &0 &0 &0 \\
0 &\underline{0} &0 &1 
\end{bmatrix}$
\end{center}
Let us note that we write the column entries having in mind the umbra notion, which means that
numbering is from left to right (as in $f$) and from bottom to top 
(so that row number $k$ from the bottom corresponds to the grey value $k$, 
with the added range for holding the possible filtering results when using a non-flat structuring element as by $b$ in the example).

\end{example}

\subsubsection*{Step 2}\label{Step:2}
We calculate $(f\oplus b)_{Um}$ by taking the linear convolution of $f_{Um}$ and $b_{Um}$, 
by using \textit{full mode} on the \textit{domain dimensions} and the \textit{range dimension}. The working of convolution of umbras is explained in the next subsection (see, Subsection \ref{Sec:UmbraConv}) 

Within our implementation, this step is sped up using Fast Fourier Transform (FFT).

\begin{example}
\label{Example:4}
Continuing Example \ref{Example:3}, we demonstrate the 
computation of $(f\oplus b)_{Um}(2)$.
We have, from Equation \eqref{Eq:Obs3}: 
\begin{eqnarray}
\lefteqn{(f\oplus b)_{Um}(2)} \nonumber\\
& = &
f_{Um}(0)\circledast_{full}b_{Um}(2)+
f_{Um}(1)\circledast_{full}b_{Um}(1)+
\nonumber\\
&&
f_{Um}(2)\circledast_{full}b_{Um}(0)+
f_{Um}(3)\circledast_{full}b_{Um}(-1)
\nonumber
\end{eqnarray}
so that
\begin{center}
$(f\oplus b)_{Um}(2)=$
$\begin{bmatrix} 
0\\
0\\
0\\
0\\
0\\
0\\
1\\
0\\
0\\
0
\end{bmatrix}$
$+ \begin{bmatrix} 
0\\
0\\
0\\
0\\
0\\
0\\
0\\
0\\
0\\
0
\end{bmatrix}$
$+ \begin{bmatrix} 
1\\
0\\
0\\
0\\
0\\
0\\
0\\
0\\
0\\
0
\end{bmatrix}$
$+ \begin{bmatrix} 
0\\
0\\
1\\
0\\
0\\
0\\
0\\
0\\
0\\
0
\end{bmatrix}$
$= \begin{bmatrix} 
1\\
0\\
1\\
0\\
0\\
0\\
1\\
0\\
0\\
0
\end{bmatrix}$
\end{center}
\end{example}

\subsubsection*{Step 3}\label{Step:3}
$(f\oplus b)$ is determined from $(f\oplus b)_{Um}$ for each $x\in F$, using the following equation
\begin{equation}
    \label{Equation:9}
    (f\oplus b)(x) =\begin{cases}
    \max\{ y \, \vert \, (f\oplus b)_{Um} (x,y) \geq 1 \} &\text{if } x\in F\oplus B\\
    0 &\text{otherwise}
    \end{cases}
\end{equation}
where, clearly, 
\begin{eqnarray}
x\not\in F\oplus B = \lbrace x \, \vert \, x-x_b\in F \, \text{ for some } \, x_b \in B \rbrace  \nonumber\\
 \Leftrightarrow 
\not\exists y : (f\oplus b)_{Um} (x,y) \geq 1
\nonumber
\end{eqnarray}

Equation \eqref{Equation:9}, in essence, computes $\delta(p_x)$, for each $x\in F$ where $p_x$ is the polynomial whose coefficients constitute the \textit{range dimension} $(f\oplus b)_{Um}(x)$. 

Thus $(f\oplus b)$ calculated by \eqref{Equation:9} is the same as defined in \eqref{Equation:6}. 

\begin{example}
\label{Example:5}
Continuing Example \ref{Example:4}, after \textbf{Step 2}, we have, 
\begin{center}
    $(f\oplus b)_{Um}=\begin{bmatrix}
0 &0 &1 &0 &0 &1 \\
0 &1 &0 &1 &1 &0 \\
0 &0 &1 &0 &1 &0 \\
0 &0 &0 &0 &0 &1 \\
1 &0 &0 &0 &0 &0 \\
0 &0 &0 &0 &1 &0 \\
0 &0 &1 &1 &0 &0 \\
0 &1 &0 &0 &0 &0 \\
1 &0 &0 &0 &0 &0 \\
\underline{0} &0 &0 &1 &0 &0 
\end{bmatrix}$
\end{center}
(where we consequently find $(f\oplus b)_{Um}(2)$ in the third column).
Using \eqref{Equation:9}, we easily obtain $(f\oplus b)$ from the matrix above by 
taking the highest row index in which the value 1 appears in each column, 
starting from bottom row equivalent to grey value zero.
Thus we get $(f\oplus b)$ $=\begin{bmatrix} \underline{5} &8&9&8&8&9\end{bmatrix}$, 
compare Example \ref{Example:2}. 

This demonstrates how the proposed method can be used to compute the \emph{exact} dilation of an image by a non-flat filter. 

\end{example}


\subsection{Convolution of Umbras}
\label{Sec:UmbraConv}

In this subsection, we justify why the convolution 
performed in \textbf{Step 2} of our method produces the desired result.  

\begin{enumerate}[i.]
    \item Products of polynomials corresponds to \textit{full mode} of convolution of it's coefficients \cite{coremen}. 
    
    Let $p,q \in \mathbb{P}$, $p(x)=\sum _{j=0}^n f_j x^j$ , $q(x)=\sum _{j=0}^m g_j x^j$.
     
    Let $f$ and $g$ be $1$-dimensional signals on $[0,n]$ and $[0,m]$ respectively, with $f[j]=f_j$, $\forall j\in [0,n]$ and $g[j]=g_j$, $\forall j\in [o,m]$. Then, we get
    
    \begin{equation}
    \label{Eq:Obs1}
        \begin{split}
            c(x) = p(x)q(x) &= \sum _{j=0}^{m+n} h_j x^j\\
            \text{where, } h_j &=\sum_{k=\max\{0,j-m\}}^{\min \{n,j\}} f[k]g[j-k]
        \end{split}
    \end{equation}
    That is,  coefficients of $c$ are given by the $1$-dimensional signal $h$, defined on $[0,m+n]$, where $h=f\circledast _{full} g$ (compare, Equations \eqref{Eq:Obs1} and \eqref{Equation:4}).
    
\item  We now focus on the convolution of \textit{umbras}, $f_{Um}$ and $b_{Um}$.  Note that, by construction (in \textbf{Step 1}), $F$ and $B$ are finite $N$-dimensional (hyper-)rectangles and $L$ is the finite set $\{0,1,\cdots l_R\}$. Therefore, the set of indices, $F\times L$ and $B \times L$ are \emph{finite} $(N+1)$-dimensional (hyper-)rectangles. Thus, we can write 

\begingroup
\allowdisplaybreaks
\begin{align}
    \lefteqn{(f_{Um}\circledast b_{Um}) [k_1,k_2\cdots k_{N+1}]}  \nonumber \\
    &= \sum _{i_1=-\infty}^{\infty} \cdots \sum _{i_{N+1}=-\infty}^{\infty} f_{Um}[i_1,i_2 \cdots i_{N+1}] \nonumber \\
    & \hspace{25ex} b_{Um}[\theta _1, \theta _2 \cdots \theta_{N+1}] \nonumber \\
    &=  \sum _{i_1=-\infty}^{\infty}\cdots \sum _{i_N=-\infty}^{\infty} \bigl\{  \sum _{i_{N+1}=-\infty}^{\infty} f_{Um}[i_1,\cdots i_{N+1}] \nonumber \\
    &\hspace{32.5ex} b_{Um}[\theta _1, \cdots \theta_{N+1}] \bigr\} \nonumber  \\
    &=  \sum _{i_1=-\infty}^{\infty}\cdots \sum _{i_N=-\infty}^{\infty} \bigl\{  (f_{Um}[i_1,\cdots i_{N},:]\circledast \nonumber  \\
     & \hspace{25ex} b_{Um}[\theta _1, \cdots \theta_{N},:])(k_{N+1}) \bigr\} \label{Eq:Obs3_0}
\end{align}

\endgroup

Here, $\theta _j = k_j-i_j$ and $f_{Um} [k_1, \ldots, k_N,:]$ and $b_{Um}[\theta_1,\ldots,\theta_{N},:]=b_{Um} [k_1-i_1,\ldots, k_{N}-i_{N},:]$ are $1$-dimensional signals, with index of every dimension, except the last, fixed.

Computing \eqref{Eq:Obs3_0} for all indices of the last dimension (\textit{range} dimension) we obtain,

\begin{equation}
    \label{Eq:Obs3}
    \hspace{-2.5ex}
    \begin{split}
      \lefteqn{(f_{Um}\circledast b_{Um}) [k_1,\cdots, k_N,:]} \\
       & =\sum _{i_1=-\infty}^{\infty}\cdots \sum _{i_N=-\infty}^{\infty} \bigl\{  (f_{Um}[i_1,\cdots i_{N},:]\circledast b_{Um}[\theta _1, \cdots \theta_{N},:]) \bigr\} 
    \end{split}
\end{equation}

\item Taking \textit{full mode} of convolution along each dimension in Equation \eqref{Eq:Obs3}, we obtain
\begin{equation}
 \label{Eq:Obs4}
    \begin{split}
         \lefteqn{(f_{Um}\circledast_{full} b_{Um}) [k_1,\cdots, k_N,:]}  \\
       & =\sum _{i_1\in F_1 \land \theta_1 \in B_1}\cdots \sum _{i_N\in F_N \land \theta_N \in B_N} \bigl\{  (f_{Um}[i_1,\cdots i_{N},:]\circledast _{full}\\
       &\hspace{35.0ex} b_{Um}[\theta _1, \cdots \theta_{N},:]) \bigr\} 
    \end{split}
\end{equation}

We have 
\begin{equation*}
\begin{split}
\lefteqn{\{i_j\vert i_j\in F_j \land \theta_j \in B_j\}}\\
&=\{i_j\vert i_j\in F_j \land k_j-i_j \in B_j\}\\
&=\{k_j-i_j\vert k_j-i_j\in F_j \land i_j \in B_j\}\\
&=\{\theta _j\vert \theta_j\in F_j \land i_j \in B_j\}
\end{split}
\end{equation*}
Thus, we can rewrite the Equation \eqref{Eq:Obs4} as

\begin{equation}
 \label{Eq:Obs5}
    \begin{split}
        \lefteqn{(f_{Um}\circledast_{full} b_{Um}) [k_1,\cdots, k_N,:]} \\
       & =\sum _{\theta_1\in F_1 \land i_1 \in B_1}\cdots \sum _{\theta_N\in F_N \land i_N \in B_N} \bigl\{  (f_{Um}[\theta_1,\cdots \theta_{N},:]\\
       & \hspace{30ex} \circledast _{full}b_{Um}[i_1, \cdots i_{N},:]) \bigr\} 
    \end{split}
\end{equation}

$F_j$ and $B_j$ are closed sub-intervals of  $\mathbb{Z}$, let $F_j=[a_1,a_2]$ and $B_j=[b_1,b_2]$. We show that the set of indices $i_j$, and thus set of indices $\theta _j=k_j-i_j$, satisfying $\theta_j\in F_j \land i_j \in B_j$ form a closed sub-interval of $\mathbb{Z}$.

\begin{align*}
    i_j \in B_j &\Rightarrow b_1\leq i_j\leq b_2\\
    k_j-i_j\in F_j &\Rightarrow a_1\leq k_j-i_j\leq a_2 \\
    &\Rightarrow k_j-a_2\leq i_j \leq k_j-a_1\\
    \therefore \max\{k_j-a_2, b_1 \} \leq & i_j \leq \min\{k_j-a_1,b_2 \}
\end{align*}
Consider a fixed $x=(k_1, k_2 \cdots k_N) \in \mathbb{Z}^N$. Let $I_j=\{i_j\vert \theta_j\in F_j \land i_j \in B_j\}$ and $\Theta _j=\{\theta_j \vert \theta_j\in F_j \land i_j \in B_j\}$, for $j=1,2,\cdots N$. $I_j$s and $\Theta _j$s are closed sub-intervals of $\mathbb{Z}$ and thus, $\prod _{j=1}^{N} I_j$ and  $\prod _{j=1}^{N} \Theta_j$ are $N$-dimensional (hyper-)rectangles.

Moreover, if $y=(i_1,i_2, \cdots, i_N)\in \prod _{j=1}^{N} I_j$, then, clearly, \\
    $x-y=(k_1-i_1, k_2-i_2, \cdots, k_N-i_N)\in \prod _{j=1}^{N} \Theta_j$
\\
Since, $F=\prod _{j=i} ^{N} F_j$ and $B=\prod _{j=i} ^{N} B_j$, we get 

\begin{equation}
\label{Eq:Obs6}
    \begin{split}
    \lefteqn{\{(x-y,y)\vert (x-y)\in F \land y\in B\}} \\
    &= \{(x-y,y)\vert (x-y)\in \prod _{j=1}^{N} \Theta_j \land y\in \prod _{j=1}^{N} I_j\}
    \end{split}
\end{equation}

From Equations \eqref{Eq:Obs5} and \eqref{Eq:Obs6}, we get

\begin{equation}
\label{Eq:Obs7}
    \hspace{-5.0ex}
    \begin{split}
    \lefteqn{(f_{Um}\circledast_{full} b_{Um}) (x)} \\
    &= (f_{Um}\circledast_{full} b_{Um}) [k_1,\cdots, k_N,:]  \\
    &= \sum _{\theta_1\in F_1 \land i_1 \in B_1}\cdots \sum _{\theta_N\in F_N \land i_N \in B_N} \bigl\{ (f_{Um}[\theta_1,\cdots \theta_{N},:]\circledast _{full}\\
       & \hspace{35ex} b_{Um}[i_1, \cdots i_{N},:]) \bigr\} \\
    &=\sum _{ (x-y)\in \prod _{j=1}^{N} \Theta_j \land y\in \prod _{j=1}^{N} I_j}  f_{Um}(x-y) \circledast _{full} b_{Um} (y)\\
    &=\sum _{ (x-y)\in F \land y\in B}  f_{Um}(x-y) \circledast _{full} b_{Um} (y)
  \end{split} 
\end{equation}
\end{enumerate}
Clearly, the above equation holds for every $x\in \mathbb{Z}^N$ such that $x-y\in F$ for some $y\in B$, i.e.\ for every $x\in F\oplus B$.
Compare Equation \eqref{Eq:Obs7} and the formula for dilation, Equation \eqref{Equation:6new}. In \eqref{Eq:Obs7}, to compute $(f_{Um}\circledast_{full} b_{Um})$ at $x$, we take sum of product of polynomials ($f_{Um}(x-y) \circledast _{full} b_{Um} (y)$) at exactly those pair of indices ($\{(x-y,y)\vert (x-y)\in F \land y\in B\} $)  on which the maximum of sum of pixel values $f(x-y)+b(y)$ is computed for finding value of dilated image $(f\oplus b)$ at $x$.

Thus, our constructions in \textbf{Step 1} and convolution in \textbf{Step 2} allows us to effectively apply Theorem \ref{t:MaxPluPoly} to compute dilation of images.

\subsection{Time Complexity}
\label{Sec:Time}

Let $n_i$ be the size of the image and $n_f$ be the size of the filter and the range of the grey-values be $[0,n_r-1]$. \textbf{Step 1} takes $\mathcal{O}(2n_in_r+2n_fn_r)$. 
For convolution using FFT 
in \textbf{Step 2}, it takes $\mathcal{O}(2n_in_r\log(2n_in_r))+\mathcal{O}2n_in_f\log(2n_in_f))$. \textbf{Step 3} takes  $\mathcal{O}(2n_in_r)$. In total, the time complexity is $\mathcal{O}(2n_in_r\log(2n_in_r)+2n_in_f\log(2n_in_f))$. 

In practice, we have $n_f\leq n_i$ and $n_r$ is a small constant, say $c$ ($c=256$ and $c=16$, in case of $8$-bit and $4$-bit grey-value images, respectively). Therefore, the overall time complexity is $\mathcal{O}(4cn_i\log(cn_i))$ $=\mathcal{O}(n_i\log(n_i))$. 

 For the classical dilation, when we use non-flat filter with no restrictions, the time complexity is $\mathcal{O}(n_f n_i)$. As $n_f\rightarrow n_i$, we get $\mathcal{O}(n_f n_i)\rightarrow \mathcal{O}(n_i ^2)$. Therefore, for large images with relatively large filter size (w.r.t. image size) or small range of pixel values (e.g., $[0,15]$ in $4$-bit image or $[0,255]$ in $8$-bit image), our proposed method is more suitable. The experiments in the following sections confirm this.

\section{Experimental Evaluation}
\label{Sec:Experiments}

In this section we first discuss the computational performance of the proposed
method in comparison to several alternative methods for dilation/erosion. 
After that we demonstrate the qualitative properties of our new exact method in comparison
with the computationally related approach based on analytic fast approximation
introduced in \cite{VMB_1}.

\subsection{Quantitative Performance Evaluation}

For our discussion, let the image be of size $n_i$ and filter of size $n_f$, given again
in terms of the number of corresponding pixels, with $n_f \leq n_i$.
Let us briefly describe all the algorithms that we compare. 

We consider the Proposed method (capitalized here for better identification)
for $4$-bit and $8$-bit tonal range, where the latter
corresponds to standard grey value range. 
Though the asymptotic time complexity, 
$\mathcal{O}(n_i\log n_i)$, remains the same for higher rates, the tonal range of image and filter 
has a significant effect on the run time of our method. For computations within our scheme, we 
have used \textsf{fft.fftn()} and \textsf{fft.ifftn()} from NumPy package \cite{NumPy}, 
for FFT and inverse FFT respectively. 
All computations are performed on a  modern workstation (Intel\textsuperscript{\textregistered}  Xeon \textsuperscript{\textregistered}  W-2125 Processor, Fedora Linux 36 (64-bit), 64GB RAM). 

The first method for comparison is a naive implementation of classical dilation
(denoted here as Classical), i.e. computing 
by sliding the filter over each pixel. 
Furthermore, we employ SciPy routine \textsf{ndimage.grey\_dilation()} from 
SciPy package \cite{SciPy} for comparison. This is a highly efficient implementation of 
the histogram sliding window described in \cite{Droogenbroeck-Talbot-1996}, using the 
approach described in \cite{Douglas-96} to compute the $\min$ and $\max$.
The run-time of naive implementation and SciPy method is independent of the tonal range 
of the image. Therefore, we only test the run-time for $8$-bit tonal range. 

Note that the worst case time complexity of naive implementation and SciPy method with 
non-flat filter, even if optimised in implementation, is theoretically still $O(n_f \times n_i)$.

\begin{figure}[h]
\centering
\begin{tikzpicture}[scale=0.95]
\begin{axis}[
    title={Dependency on Filter Size},
    xlabel={Filter Size [$n_f=n_1\times n_1$]},
    ylabel={Time [in seconds]},
    xmin=512, xmax=65536,
    ymin=-2.5, ymax=28,
    xtick={},
    ytick={},
    legend pos=north west,
    ymajorgrids=true,
    grid style=dashed,
]

\addplot[
    color=blue,
    mark=square,
    ]
    coordinates {
    (1024,0.7563)(4096,0.7599)(9216,0.7601)(16384,0.7579)(25600,0.7548)(36864,0.7610)(50176,0.7566)(65536,0.7555)
    };
    \addlegendentry{Proposed method $4$-bit}

\addplot[
    color=blue,
    mark=triangle,
    ]
    coordinates {
    (1024,12.706)(4096,12.7066)(9216,12.7072)(16384,12.7053)(25600,12.7104)(36864,12.7077)(50176,12.7093)(65536,12.7064)
    };
    \addlegendentry{Proposed method $8$-bit}
 
\addplot[
    color=red,
    mark=square,
    ]
    coordinates {
    (1024,0.2433)(4096,1.0522)(9216,2.6518)(16384,5.5421)(25600,10.528)(36864,15.1803)(50176,20.6720)(65536,27.0016)
    };
    \addlegendentry{SciPy Dilation}
 
\addplot[
    color=black,
    mark=triangle,
    ]
    coordinates {
    (1024,0.4199)(4096,1.5696)(9216,3.4621)(16384,5.8188)(25600,8.9975)(36864,12.9564)(50176,17.6351)(65536,23.0336)
    };
    \addlegendentry{Classical Dilation}

\end{axis}
\end{tikzpicture}
\begin{tikzpicture}[scale=0.95]
\begin{axis}[
    title={Dependency on Image Size},
    xlabel={Image Size [$n_i=n\times n$]},
    ylabel={Time [in seconds]},
    xmin=16384, xmax=16777216,
    ymin=-200, ymax=4500,
    xtick={},
    ytick={},
    legend pos=north west,
    ymajorgrids=true,
    grid style=dashed,
]

\addplot[
    color=blue,
    mark=square,
    ]
    coordinates {
    (16384,0.428)(65536,0.192)(262144,0.795)(1048576,3.361)(4194304,15.765)(16777216,65.47)
     };
    \addlegendentry{Proposed method $4$-bit}

\addplot[
    color=blue,
    mark=triangle,
    ]
    coordinates {
    (16384,0.709)(65536,3.027)(262144,13.176)(1048576,55.687)(4194304,261.21)(16777216,1084.8)
    };
    \addlegendentry{Proposed method $8$-bit}
  
\addplot[
    color=red,
    mark=square,
    ]
    coordinates {
    (16384,0.0035)(65536,0.0385)(262144,0.6532)(1048576,10.1979)(4194304,186.46)(16777216,3031.36)
    };
    \addlegendentry{SciPy Dilation}
 
\addplot[
    color=black,
    mark=triangle,
    ]
    coordinates {
    (16384,0.0039)(65536,0.0674)(262144,1.0383)(1048576,16.085)(4194304,273.0129)(16777216,4368.21)
    };
    \addlegendentry{Classical Dilation}

\end{axis}
\end{tikzpicture}
\caption{\label{figure-complexity}
Evaluation of algorithmic time complexity of our method.
{\bf Top:} Varying filter sizes on an image of size $512\times512$. 
Let us note that the lower axis is given in factors of $10^4$. 
{\bf Bottom:} Varying image sizes $n_i=n \times n$, with filter size 
$n_f = \lfloor\frac{n}{10}\rfloor\times\lfloor\frac{n}{10}\rfloor$. 
Let us note that the lower axis is given in factors of $10^7$. }
\end{figure}
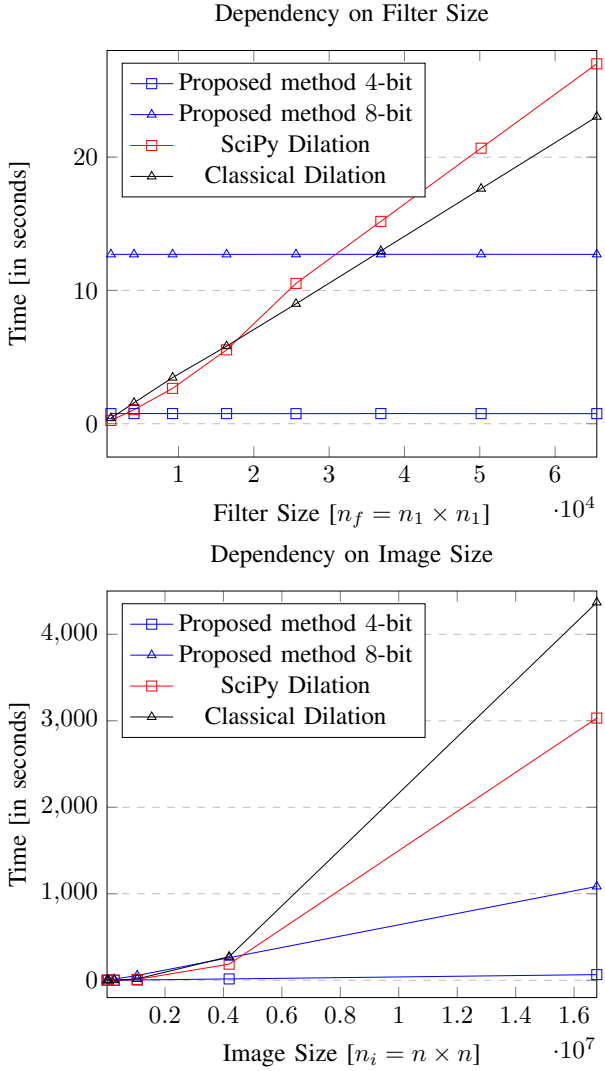

\begin{table*}
  \caption{FPGA Resource Utilization}
  \label{tab:FPGA1}
  \centering
  \begin{tabular}{@{}|c|c||c|c|c|c|c|@{}}
    \hline
     Tonal Range & FFT Size & BRAM & DSP & FF & LUT & URAM \\
    \hline
    \multirow{3}{*}{2-bit tonal range} & 32 & 112 (5\%) & 10 (~0\%) & 29636 (1\%) & 29366 (3\%)& 2 (~0\%)\\
                                 & 64 & 68 (3\%) & 10 (~0\%) & 33900 (1\%) & 33246 (3\%)&50 (10\%)\\
                                 & 128 & 116 (5\%) & 10 (~0\%) & 38736 (2\%) & 44075 (4\%)&50 (10\%)\\
                                 & 256 & 392 (20\%) & 10 (~0\%) & 39772 (2\%) & 42415 (4\%)&50 (10\%)\\
                                 & 512 & 650 (33\%) & 10 (~0\%) & 43572 (2\%) & 50655 (5\%)&50 (10\%)\\
                                 & 1024 & 1194 (61\%) & 10 (~0\%) & 47382 (2\%) & 61894 (6\%)&50 (10\%)\\
    \hline
    \multirow{3}{*}{3-bit tonal range} & 32 & 208 (10\%) & 10 (~0\%) & 31817 (1\%) & 35821 (3\%) & 2 (~0\%)\\
                                 & 64 & 116 (5\%) & 10 (~0\%) & 36137 (2\%) & 38645 (4\%) & 98 (21\%)\\
                                 & 128 & 212 (10\%) & 10 (~0\%) & 41077 (2\%) & 49928 (5\%) & 98 (21\%)\\
                                 & 256 & 488 (25\%) & 10 (~0\%) & 42217 (2\%) & 48324 (5\%) & 98 (21\%)\\
                                 & 512 & 746 (38\%) & 10 (~0\%) & 46121 (2\%) & 57126 (6\%) & 98 (21\%)\\
                                 & 1024 & 1290 (66\%) & 10 (~0\%) & 50050 (2\%) & 68711 (7\%) & 98 (21\%)\\
    \hline
    \multirow{3}{*}{4-bit tonal range} & 32 & 400 (20\%) & 10 (~0\%) & 39752 (2\%) & 52866 (5\%) & 2 (~0\%)\\
                                 & 64 & 212 (10\%) & 10 (~0\%) & 44195 (2\%) & 53578 (5\%) & 194 (41\%)\\
                                 & 128 & 404 (20\%) & 10 (~0\%) & 49332 (2\%) & 65769 (7\%) & 194 (41\%)\\
                                 & 256 & 680 (35\%) & 10 (~0\%) & 50693 (2\%) & 64845 (7\%) & 194 (41\%)\\
                                 & 512 & 938 (48\%) & 10 (~0\%) & 54806 (3\%) & 74219 (8\%) & 194 (41\%)\\
                                 & 1024 & 1482 (76\%) & 10 (~0\%) & 58929 (3\%) & 86470 (9\%) & 194 (41\%)\\
    \hline
    \multirow{3}{*}{5-bit tonal range} & 32 & 788 (40\%) & 10 (~0\%) & 51967 (2\%) & 82938 (9\%) & 2 (~0\%)\\
                                 & 64 & 408 (21\%) & 10 (~0\%) & 56623 (3\%) & 79420 (8\%) & 386(83\%)\\
                                 & 128 & 792 (40\%) & 10 (~0\%) & 62187 (3\%) & 93431 (10\%) & 386(83\%)\\
                                 & 256 & 1068 (55\%) & 10 (~0\%) & 63964 (3\%) & 93857 (10\%) & 386(83\%)\\
                                 & 512 & 1326 (68\%) & 10 (~0\%) & 68493 (3\%) & 104387 (11\%) & 386(83\%)\\
                                 & 1024 & 1870 (96\%) & 10 (~0\%) & 73032 (4\%) & 117996 (13\%) & 386(83\%)\\
    \hline
  \end{tabular}
\end{table*}

\begin{figure}[h]
\centering
\begin{tikzpicture}[scale=0.9]
\begin{axis}[
    title={FPGA Implementation Execution Time},
    xlabel={1-D FFT Size},
    ylabel={Time [in seconds]},
    xmin=32, xmax=1024,
    ymin=-1.5, ymax=30,
    xtick={},
    ymode=log,
    ytickten={-3,-2,-1,0,1,2},
    legend pos=south east,
    ymajorgrids=true,
    grid style=dashed,
]

\addplot[
    color=blue,
    mark=square,
    ]
    coordinates {
    (32,0.00502)(64,0.01835)(128,0.06991)(256,0.271)(512,1.069)(1024,4.237)
    };
    \addlegendentry{Proposed method $2$-bit}

\addplot[
    color=red,
    mark=square,
    ]
    coordinates {
    (32,0.00909)(64,0.0329)(128,0.125)(256,0.482)(512,1.894)(1024,7.499)
    };
    \addlegendentry{Proposed method $3$-bit}
\addplot[
    color=green,
    mark=square,
    ]
    coordinates {
    (32,0.0172)(64,0.06214)(128,0.23)(256,0.904)(512,3.551)(1024,14.044)
    };
    \addlegendentry{Proposed method $4$-bit}
\addplot[
    color=black,
    mark=square,
    ]
    coordinates {
    (32,0.0326)(64,0.116)(128,0.438)(256,1.683)(512,6.6)(1024,26.086)
    };
    \addlegendentry{Proposed method $5$-bit}

\end{axis}
\end{tikzpicture}
\caption{\label{figure-fpgaexec}
Execution time of the proposed method with different tonal ranges and for different image sizes}
\end{figure}
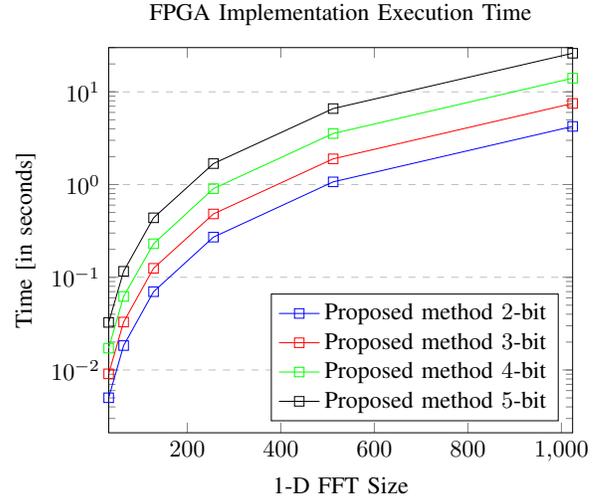

In the first experiment, see Figure \ref{figure-complexity} \textit{top}, we evaluate 
the average time taken by varying the filter size on a fixed size of image. 
The filters and images are filters generated using \textsf{numpy.random.randint()}, with 
range of values from $0$ to $255$ for $8$-bit Proposed, Classical and SciPy method 
and with range of values from $0$ to $15$ for $4$-bit Proposed method. 
The size of images are $512\times 512$. It is clearly visible that the SciPy method 
and the Classical dilation behave linearly with respect to size of filter.
For filter size $32\times 32$ they take about $0.24$ and $0.42$ seconds absolute 
computation time on our computer, respectively, up to around $25$ seconds for filter 
size $256 \times 256$. The time taken for Proposed method remains constant with the 
size of the filter, taking on average about $0.76$ seconds in $4$-bit settings and 
about $12.5$ in $8$-bit setting.

In the second experiment, Figure \ref{figure-complexity} \textit{bottom}, we evaluate 
the time taken for dilation on varying sizes of images. 
The image sizes$n_i=n\times n$ increases from $128\times 128$ to $4096\times 4096$. 
The corresponding filter sizes are $n_f= \lfloor\frac{n}{10}\rfloor\times\lfloor\frac{n}{10}\rfloor$.
The images and the filters are generated using \textsf{numpy.random.randint()}, with range of 
values from $0$ to $255$, except for $4$-bit Proposed method, where the range of values 
is $0$ to $15$. As expected, the Proposed method in $4$-bit and $8$-bit settings 
perform in $\mathcal{O}(n_i\log n_i)$ time. In $4$-bit setting, the Proposed 
method takes $0.42$ seconds for dilation of $128\times 128$ image by $12\times 12$ filter, 
and $15$ seconds for dilation of $2048\times 2048$ image by a $200\times 200$ filter. 
In $8$-bit settings, the Proposed method takes $0.7$ and $261.2$ seconds, respectively. 
We have, in the second experiment, 
$n_f=\lfloor\frac{n}{10}\rfloor\times\lfloor\frac{n}{10}\rfloor \approx	 \frac{n_i}{100}$. 
Thus the time complexity of SciPy method and naive method is $\mathcal{O}(n_in_f)=\mathcal{O}(n_i^2)$.
This is also reflected by their run-times in the second experiment.

We observe from the above experiments that our Proposed method is significantly faster 
than the other methods in the narrow tonal range, as e.g. in the $4$-bit setting. The Proposed 
method is faster than SciPy method and Classical method, also in the usual $8$-bit setting, 
if the ratio of filter size to image size is in the larger range, as seen in first experiment 
or when working on larger images, even keeping the ratio of filter size to image size 
constant, as in second experiment.

Let us note that we have not employed GPUs in the above experiments. 
It is surely worth pointing out that attempts to use GPUs to compute grey value 
morphology usually incorporates restrictions on symmetry and/or flatness of 
the filter, see discussions in \cite{Moreaud,Thurley-Danell-2012}.
However, there are several efficient implementations of FFT and inverse FFT 
on the GPU, see e.g. \cite{GPUFFT}. Therefore, our method could be sped up 
utilising the GPUs, without any restrictions on the filter, making it even 
more competitive in the large tonal range.  

\subsection{Quantitative Results of the Hardware Implementation}
In this work, we accelerated the proposed method by implementing it on hardware and observing the results. We used a modern Xilinx FPGA board, the Versal VCK190 Evaluation Board, which houses an XCVC1902-2M FPGA device, to implement a hardware representation of our method. The IP core designed for this method was written in C++ and synthesized using Xilinx Vitis HLS 2022.2. The results from the hardware implementation were compared to the Python code to ensure the functionality of the design. For our hardware tests, we chose images and filters as squares with sizes of each edge in a manner, to form padded images and filters with edge sizes that are powers of 2 to fit the FFT cores optimally. We used Xilinx FFT IP core for 1-dimensional FFTs (forward and inverse) to eventually implement 3-dimensional FFTs.

We tested our method for tonal ranges of 2-bit, 3-bit, 4-bit, and 5-bit, with the filter size set to $5\times5$, and changed the image size from $28\times28$ up to $1020\times1020$. The FPGA resource utilization summary is shown in Table \ref{tab:FPGA1} for these scenarios. It can be seen that the highest resource utilization is in memory blocks, which are used to store the contents of 2-dimensional and 3-dimensional arrays at different points of the procedure. Choosing a higher tonal range limits the largest image size that can be processed, considering the BRAM and URAM resources available on the FPGA.

Figure \ref{figure-fpgaexec} shows the execution time for the four tonal range scenarios and different FFT sizes. The FFT size addresses the size of the 1-D FFTs processing the x-axis and y-axis of the padded image and filter. The execution time data shows that, in each scenario, when the FFT size doubles, the execution time almost doubles as well, which is consistent with the results from the Python implementation of this method.
\subsection{Qualitative Comparison to Previous Fourier Approach}


The significant advantage of our method is the fact that we are able to compute the 
\emph{exact} dilation of an image of size $n_i$ by any non-flat filter of size 
$n_f\leq n_i$ in $\mathcal{O}(n_i\log n_i)$ time. 
The Fast Analytical Approximation proposed in \cite{VMB_1}, is also asymptotically 
$\mathcal{O}(n_i\log n_i)$ as the new Proposed method and, in fact, has faster 
run-times in practice. The downside of Fast Approximation method is the 
non-trivial grey-value shift, that has been studied in detail in \cite{VMB_2}, 
which comes along with a smoothing effect in the tonal histogram.


\begin{figure}[h]
\centering
\minipage{0.33\linewidth}
\centering
 \includegraphics[width=\linewidth]{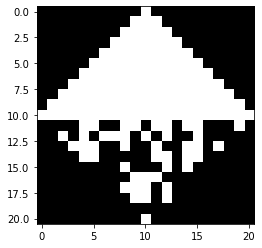}
\endminipage\hfill
\minipage{0.66\linewidth}
 \includegraphics[width=\linewidth]{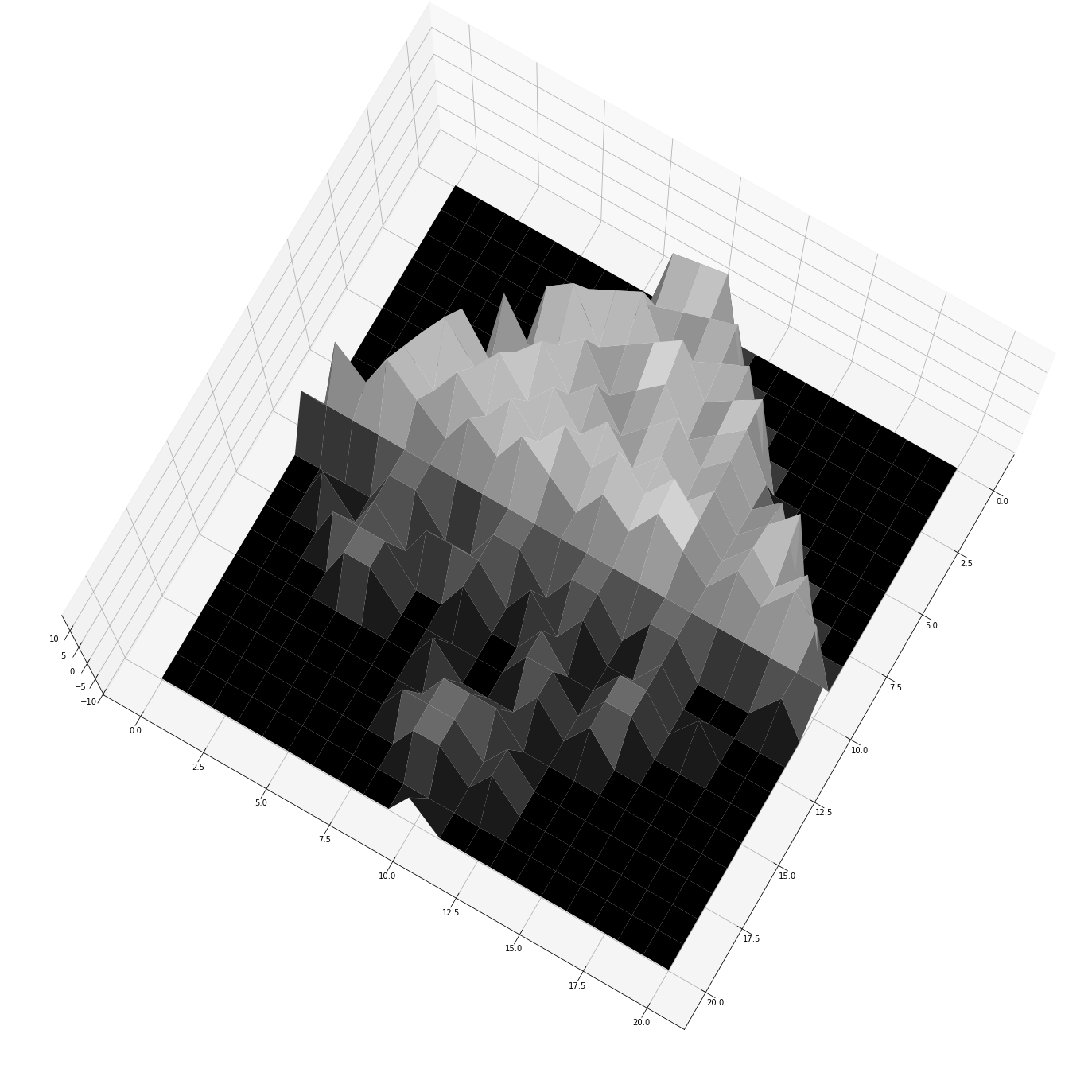}
\endminipage\hfill
\caption{\label{FilterShapeUmbra} 
{\bf: Left} The base shape of filter.  {\bf Right:} Umbra of the filter of image.
}

\end{figure}

To demonstrate the exactness of the Proposed method and its difference from 
Fast Approximation, we use a non-flat filter. 
The shape of the filter and its umbra (again after \cite{Haralick_1}) is 
shown in Figure \ref{FilterShapeUmbra}. 

We perform the demonstration of dilation quality on a $512\times 512$ image of 
\textit{Peppers}, see Figure \ref{pepper} (left). 
The Figure \ref{pepper} (right) shows the dilation of the \textit{Pepper} 
image with non-flat filter in the classical way, i.e. as described in \eqref{Equation:6}.


\begin{figure}[h]
\centering
\minipage{0.48\linewidth}
\centering
 \includegraphics[width=\linewidth]{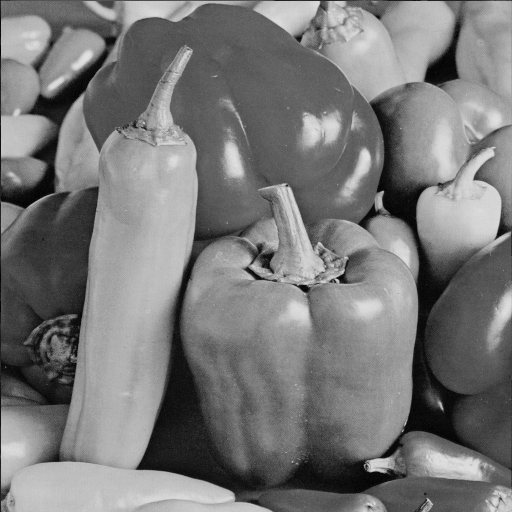}
\endminipage\hfill
\minipage{0.48\linewidth}
 \includegraphics[width=\linewidth]{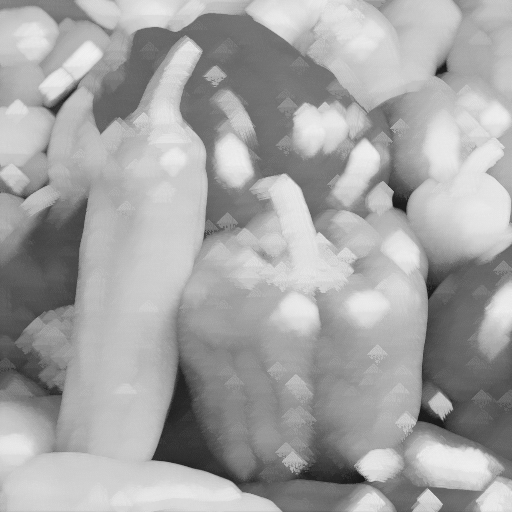}
\endminipage\hfill
\caption{\label{pepper} 
{\bf: Left:}  \textit{Pepper} image of size $512\times 512$ {\bf Right:} Classical Dilation with the non-flat filter (as described in Figure \ref{FilterShapeUmbra}.
}

\end{figure}


\begin{figure}[h]
\centering
\minipage{0.48\linewidth}
\centering
 \includegraphics[width=\linewidth]{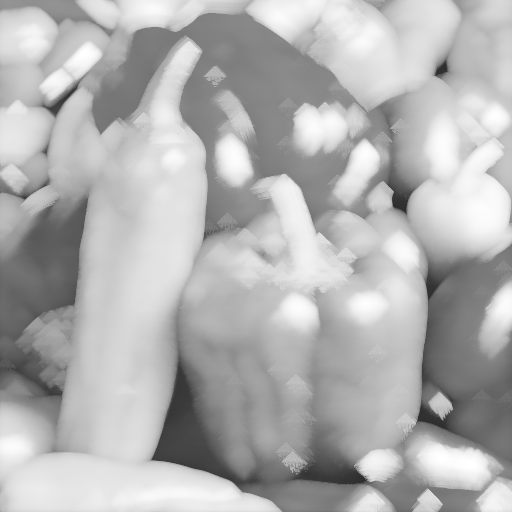}
\endminipage\hfill
\minipage{0.48\linewidth}
 \includegraphics[width=\linewidth]{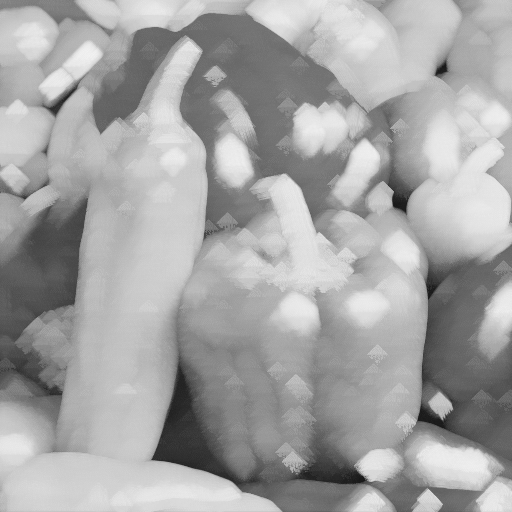}
\endminipage\hfill
\caption{\label{compDilImg} 
Dilation of  \textit{Pepper} image of size $512\times 512$ with the non-flat filter {\bf Left:} Fast Approximation. {\bf Right:} Proposed Method.
}

\end{figure}


\begin{figure}[h]
\centering
\minipage{0.48\linewidth}
\centering
 \includegraphics[width=\linewidth]{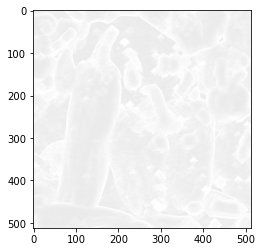}
\endminipage\hfill
\minipage{0.48\linewidth}
 \includegraphics[width=\linewidth]{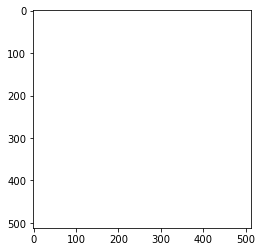}
\endminipage\hfill
\caption{\label{DiffImage} 
Negative of absolute difference (pixel-wise) with Classical Dilation of  \textit{Pepper} image with the non-flat filter {\bf Left:} Fast Approximation. {\bf Right:} Proposed Method.
}

\end{figure}

We compare the Proposed method (in $8$-bit settings) and Fast 
Approximation, see Figure \ref{compDilImg}, with the classical exact dilation.

The \emph{exactness} of the Proposed method is confirmed visually by looking 
at the images Figure \ref{pepper} (right) and  Figure \ref{compDilImg} (right). This is quantitatively confirmed, in Figure \ref{DiffImage}, by taking the absolute difference, pixel-by-pixel, of the result of classical dilation with that of fast approximation and proposed method.
It is also evident by overlaying the histograms of the images, 
see Figure \ref{HistoCompUmbra}, that there are
no artefacts or dilation accuracy degradation by the boundary treatment inside  the FFT.

The positive grey-value shift in Fast Approximation is as expected apparent in 
Figure \ref{HistoCompUmbra}. Moreover, our proposed method is also successful in 
preserving the minute details which are lost in Fast Approximation, compare 
images in Figure \ref{compDilImg}.


\begin{figure}[h]

 \includegraphics[width=\linewidth]{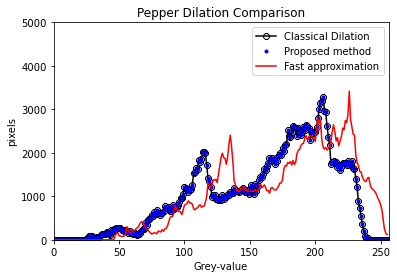}

\caption{\label{HistoCompUmbra} 
Histogram of dilated images.
}

\end{figure}

\begin{figure}[h]
\centering
\begin{tikzpicture}[scale=0.95]
\begin{axis}[
    title={Dependency on Filter Size},
    xlabel={Filter Size [$n_f=n_1\times n_1$]},
    ylabel={Average Absolute Error},
    xmin=0, xmax=320,
    ymin=-2, ymax=20,
    xtick={},
    ytick={},
    legend pos=north west,
    ymajorgrids=true,
    grid style=dashed,
]

\addplot[
    color=blue,
    mark=square,
    ]
    coordinates {
    (9,0)(25,0)(49,0)(81,0)(121,0)(169,0)(225,0)(289,0)
     };
    \addlegendentry{Proposed Method}

\addplot[
    color=red,
    mark=triangle,
    ]
    coordinates {
    (9,4.861)(25,7.66)(49,10.280)(81,12.340)(121,14.383)(169,16.008)(225,17.52)(289,18.77)
    };
    \addlegendentry{Fast Approximation}

\end{axis}
\end{tikzpicture}
\begin{tikzpicture}[scale=0.95]
\begin{axis}[
    title={Dependency on Image Size},
    xlabel={Image Size [$n_i=n\times n$]},
    ylabel={Average Absolute Error},
    xmin=5000, xmax=365000,
    ymin=-2, ymax=15,
    xtick={},
    ytick={},
    legend pos=north west,
    ymajorgrids=true,
    grid style=dashed,
]

\addplot[
    color=blue,
    mark=square,
    ]
    coordinates {
    (10000,0)(40000,0)(90000,0)(160000,0)(250000,0)(360000,0)
     };
    \addlegendentry{Proposed Method}

\addplot[
    color=red,
    mark=triangle,
    ]
    coordinates {
    (10000,7.7)(40000,7.7)(90000,7.7)(160000,7.7)(250000,7.7)(360000,7.7)
    };
    \addlegendentry{Fast Approximation}

\end{axis}
\end{tikzpicture}

\caption{\label{Fig:ErrorGraph}
Average \emph{absolute} difference in pixel value with Classical 
{\bf Top:} Varying filter sizes $n_f=n_1 \times n_1$, on an image of size $99\times 99$. 
{\bf Bottom:} Varying image sizes $n_i=n \times n$, with filter size 
$5 \times 5$. 
}
\end{figure}
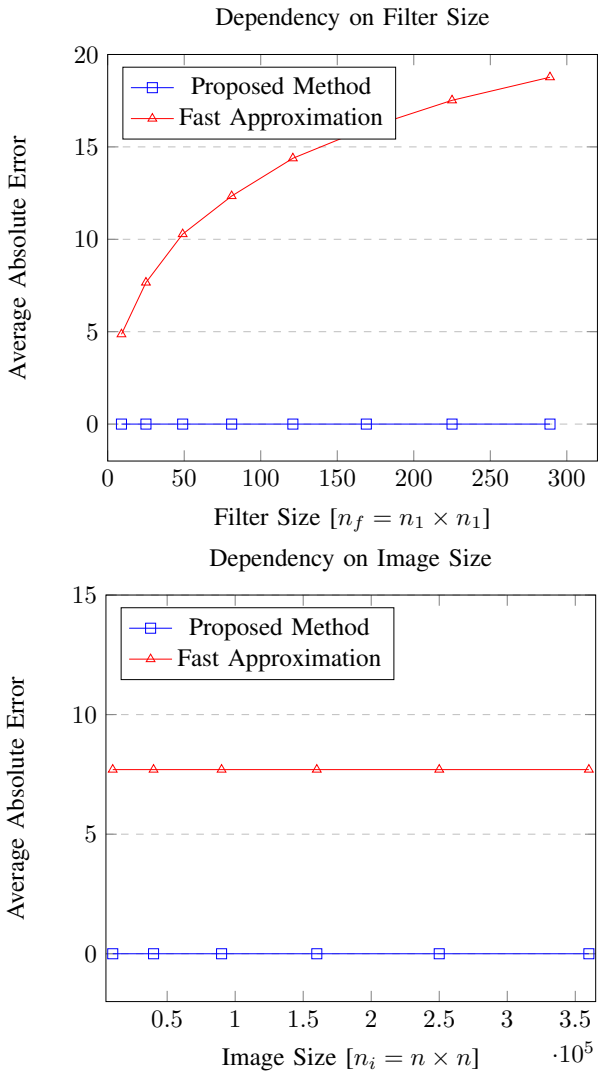

To demonstrate the exactness the proposed method, we perform two experiments, see Figure \ref{Fig:ErrorGraph}, in $8$-bit settings. We measure the average absolute difference in pixel value with classical dilation with  fast approximation and the proposed method. In the first experiment, we generate,  $100$ pair of random $99\times 99$ images and filter, for each filter sizes $3\times 3$, $5\times 5$ $\cdots$ $17\times 17$, using  \textsf{numpy.random.randint()}, with 
range of values from $0$ to $255$.Similarly, in the second experiment, we generate,  $100$ pair of random $5\times 5$ filters and images, for each image sizes $100\times 100$, $200\times 200$ $\cdots$ $600\times 600$. As expected the Proposed method has $0$ average absolute error, i.e.\ it exactly equates with classical dilation.
We can also observe that the absolute error in Fast Approximation increases logarithmically with respect to filter size.

\section{Conclusion}
\label{Sec:Conclusion}
We have proposed a novel method to \emph{exactly} compute morphological dilation of an 
image, of size $n_i$, with any arbitrary non-flat filter, of size $n_f\leq n_i$, 
in $\mathcal{O}(n_i \log n_i)$ time. As erosion is dual to dilation, it is straightforward
to compute erosion analogously.

Because our novel scheme is exact, any useful morphological filter combinations like 
opening, closing, Beucher gradient, and so on, can be easily combined in a novel and fast way.
We conjecture that this may make our algorithm a novel and useful basis for many practical
applications.

 We have established homomorphism between the max-plus semi-ring of non-negative integers and a semi-ring of polynomials set in the real field. This theory could serve as a foundation to future research relating max-plus semi-rings and plus-prod semi-rings (see e.g. \cite{Cohen-2019}, \cite{Maragos_2019}).  

Our proposed method allows us to  explore  some of well-engineered techniques developed 
for computing convolution in future work. For example, heading for real time applications in very large 
images, say $n_i\geq 10^9$, some partitioned convolution algorithm 
\cite{Wefers-2015} can be implemented. 
We may improve the run-time of our implementation by employing GPUs to calculate FFT 
and inverse FFT, see e.g. \cite{GPUFFT}, and thus speeding up \textbf{Step 2} of our method. 

The $(N+1)$-dimensional arrays $f_{Um}$ and $b_{Um}$ constructed in Step 1,
see \eqref{Equation:8} and \eqref{Equation:9}, satisfy the sparsity conditions 
mentioned in \cite{Sparse-FFT-Conv}. Therefore, the performance, especially for broader tonal ranges, may also be improved by using 
Sparse-FFT to compute the convolution in Step 2, see \cite{Sparse-FFT-Conv}.

As we have shown, the new method appears to be particularly useful for narrow tonal ranges,
but at the same time it is evident that the method may be highly efficient for 
standard tonal ranges when exploring advanced computational techniques and hardware
as mentioned. By the presented FPGA implementation we highlight 
the potential usefulness of the new method for many possible applications.


\begin{thebibliography}{1}
\bibliographystyle{IEEEtran}

\bibitem{NTT-Conv}
Agarwal, R.C. and Burrus, C.S., 1975. Number theoretic transforms to implement fast digital convolution. Proceedings of the IEEE, 63(4), pp.550-560.

\bibitem{NumPy}
Harris, C.R., Millman, K.J., Van Der Walt, S.J., Gommers, R., Virtanen, P., Cournapeau, D., Wieser, E., Taylor, J., Berg, S., Smith, N.J. and Kern, R., 2020. Array programming with NumPy. Nature, 585(7825), pp.357-362.

\bibitem{Sparse-FFT-Conv}Bringmann, Karl, Nick Fischer, and Vasileios Nakos. "Sparse nonnegative convolution is equivalent to dense nonnegative convolution." In Proceedings of the 53rd Annual ACM SIGACT Symposium on Theory of Computing, pp. 1711-1724. 2021.


\bibitem{Wefers-2015}Wefers, Frank. Partitioned convolution algorithms for real-time auralization. Vol. 20. Logos Verlag Berlin GmbH, 2015.

\bibitem{Douglas-96}
Douglas, S.C., Running max/min calculation using a pruned ordered list. IEEE Transactions on Signal Processing, 44(11), pp.~2872-2877. 1996.

\bibitem{TMG}
Tuzikov, A.V., Margolin, G.L. and Grenov, A.I., 1997. Convex set symmetry measurement via Minkowski addition. Journal of Mathematical Imaging and Vision, 7(1), pp.53-68.



\bibitem{VMB_1}
Kahra, M., Sridhar, V. and Breuß, M., 2021, May. Fast morphological dilation and erosion for grey scale images using the Fourier transform. In International Conference on Scale Space and Variational Methods in Computer Vision (pp. 65-77). Springer, Cham.


\bibitem{VMB_2}
Sridhar, V., Breuß, M. and Kahra, M., 2021, October. Fast Approximation of Color Morphology. In International Symposium on Visual Computing (pp. 488-499). Springer, Cham.

\bibitem{SciPy_Convolve}
SciPy Documentation, \url{https://docs.scipy.org/doc/scipy/ reference/generated/scipy.signal.convolve.html}. Last accessed 2 Feb 2021.





\bibitem{Serra-Soille}
Serra, J. and Soille, P. eds., 2012. Mathematical morphology and its applications to image processing (Vol. 2). Springer Science and Business Media.

\bibitem{Najman-Talbot}
Najman, L. and Talbot, H. eds., 2013. Mathematical morphology: from theory to applications. John Wiley and Sons.




\bibitem{Roerdink-2011}
Roerdink, J.B., 2011, July. Mathematical morphology in computer graphics, scientific visualization and visual exploration. In International Symposium on Mathematical Morphology and Its Applications to Signal and Image Processing (pp. 367-380). Springer, Berlin, Heidelberg.


\bibitem{Kukal}
Kukal, J., Majerová, D. and Procházka, A., 2007. Dilation and erosion of gray images with spherical masks. In the Proceedings of the 15th Annual Conference Technical Computing

\bibitem{deforges-etal-2013}
Déforges, O., Normand, N. and Babel, M., 2013. Fast recursive grayscale morphology operators: from the algorithm to the pipeline architecture. Journal of Real-Time Image Processing, 8(2), pp.143-152.

\bibitem{Moreaud}
Moreaud, M. and Itthirad, F., 2014, April. Fast algorithm for dilation and erosion using arbitrary flat structuring element: Improvement of urbach and wilkinson's algorithm to GPU computing. In 2014 International Conference on Multimedia Computing and Systems (ICMCS) (pp. 289-294). IEEE.

\bibitem{Lin-Xu-2009}
Lin, X. and Xu, Z., 2009, May. A fast algorithm for erosion and dilation in mathematical morphology. In 2009 WRI World Congress on Software Engineering (Vol. 2, pp. 185-188). IEEE.

\bibitem{herk-1992}
Van Herk, M., 1992. A fast algorithm for local minimum and maximum filters on rectangular and octagonal kernels. Pattern Recognition Letters, 13(7), pp.517-521.


\bibitem{Haralick_1} 
Haralick, R.M., Sternberg, S.R. and Zhuang, X., 1987. Image analysis using mathematical morphology. IEEE transactions on pattern analysis and machine intelligence, (4), pp.532-550.

\bibitem{coremen} Cormen, Thomas H., Charles E. Leiserson, Ronald L. Rivest, and Clifford Stein. Introduction to algorithms. MIT press, 2022.


\bibitem{Viv_1}
Sridhar, Vivek, and Michael Breuß. "Sampling of Non-flat Morphology for Grey Value Images." In International Conference on Computer Analysis of Images and Patterns, pp. 88-97. Springer, Cham, 2021.

\bibitem{Maragos_2019}Maragos, Petros. "Tropical geometry, mathematical morphology and weighted lattices." In International Symposium on Mathematical Morphology and Its Applications to Signal and Image Processing, pp. 3-15. Springer, Cham, 2019.

\bibitem{Golan_2013}Golan, Jonathan S. Semirings and their Applications. Springer Science and Business Media, 2013.

\bibitem{Herstein_1975}Herstein, I. N. (1975). Topics in Algebra. Wiley. ISBN 0471010901.

\bibitem{Thurley-Danell-2012}
Thurley, M.J. and Danell, V., 2012. Fast morphological image processing open-source extensions for GPU processing with CUDA. IEEE journal of selected topics in signal processing, 6(7), pp.849-855.


\bibitem{SciPy}
P. Virtanen et al.: SciPy 1.0: Fundamental Algorithms for Scientific Computing in Python.  Nature Methods {\bf 17}(3), 261--272, (2020)


\bibitem{Droogenbroeck-Talbot-1996}
Van Droogenbroeck, M. and Talbot, H., 1996. Fast computation of morphological operations with arbitrary structuring elements. Pattern recognition letters, 17(14), pp.1451-1460.

\bibitem{Droogenbroeck-Buckley-2005}
Van Droogenbroeck, M. and Buckley, M.J., 2005. Morphological erosions and openings: fast algorithms based on anchors. Journal of Mathematical Imaging and Vision, 22(2), pp.121-142.


\bibitem{GPUFFT} 
Moreland, K. and Angel, E., The FFT on a GPU.  In Proceedings of the ACM SIGGRAPH/Eurographics conference on Graphics hardware, pp.~112-119. 2003.

\bibitem{Jones-1999}
Jones, R., 1999. Connected filtering and segmentation using component trees. Computer Vision and Image Understanding, 75(3), pp.215-228.











\bibitem{FNT}
Agarwal, R.C. and Burrus, C., 1974. Fast convolution using Fermat number transforms with applications to digital filtering. IEEE Transactions on Acoustics, Speech, and Signal Processing, 22(2), pp.87-97.


\bibitem{SB-2022-1}
Sridhar, Vivek, and Michael Breuß. "An Exact Fast Fourier Method for Morphological Dilation and Erosion Using the Umbra Technique." In 2022 19th Conference on Robots and Vision (CRV), pp. 190-196. IEEE, 2022.

\bibitem{Pollard}
Pollard, J.M., 1971. The fast Fourier transform in a finite field. Mathematics of computation, 25(114), pp.365-374.

\bibitem{Rader}
Rader, C.M., 1972, January. The number theoretic DFT and exact discrete convolution. In IEEE Arden House Workshop on digital signal processing. Harriman NY.


\bibitem{Krizek}Krizek, M., Luca, F. and Somer, L., 2002. 17 Lectures on Fermat numbers: from number theory to geometry. Springer Science and Business Media.

\bibitem{Cohen-2019}Cohen, Guy, Stéphane Gaubert, and Jean-Pierre Quadrat. "Max-plus algebra and system theory: where we are and where to go now." Annual reviews in control 23 (1999): 207-219.


\end{thebibliography}
\end{document}